\documentclass[a4paper,12pt]{article}
\pdfoutput=1

\usepackage{geometry}
\usepackage[margin=10mm]{caption}

\usepackage[utf8]{inputenc}
\usepackage[english]{babel}
\usepackage{amssymb, amsmath}
\usepackage{amsthm}
\usepackage{eucal}
\usepackage{mathalpha}
\usepackage{amsfonts}
\usepackage{mathrsfs}
\usepackage{setspace}
\usepackage{graphicx}
\usepackage{subcaption}
\usepackage{slashed}
\usepackage{bm}
\usepackage[sorting=none, backend=biber]{biblatex}
\usepackage{array,booktabs}
\usepackage{lmodern}
\usepackage{tikz}
\usepackage{microtype}
\usepackage{hyperref}
\usepackage{authblk}
\usepackage{yfonts}
\usepackage{csquotes}

\usepackage{tikz-cd}
\usetikzlibrary{shapes.geometric}
\usepackage{caption}
\usetikzlibrary{calc}
\def\centerarc[#1](#2)(#3:#4:#5)
    { \draw[#1] ($(#2)+({#5*cos(#3)},{#5*sin(#3)})$) arc (#3:#4:#5); }

\theoremstyle{definition}
\newtheorem{definition}{Definition}[section]
\newtheorem{remark}{Remark}[section]

\newtheorem{theorem}{Theorem}
\newtheorem{prop}{Proposition}[section]
\newtheorem{lemma}{Lemma}[section]
\newtheorem{corollary}{Corollary}[section]
\newtheorem{example}{Example}[section]

\makeatletter
\renewenvironment{proof}[1][\proofname]{%
   \par\pushQED{\qed}\normalfont%
   \topsep6\p@\@plus6\p@\relax
   \trivlist\item[\hskip\labelsep\bfseries#1\@addpunct{.}]%
   \ignorespaces
}{%
   \popQED\endtrivlist\@endpefalse
}

\newcommand{\beq}{\begin{equation}}
\newcommand{\eeq}{\end{equation}}
\newcommand{\eeeem}{\end{multline}}
\newcommand{\bem}{\begin{multline}}
\newcommand{\bqa} {\begin{eqnarray}}
\newcommand{\eqa} {\end{eqnarray}}
\newcommand{\bmul}{\begin{multline}}
\newcommand{\emul}{\end{multline}}

\DeclareMathOperator{\Id}{Id}

\DeclareMathOperator{\diam}{diam}
\DeclareMathOperator{\Ad}{Ad}
\DeclareMathOperator{\ad}{ad}

\DeclareMathOperator{\dist}{dist}

\newcommand{\ups}{\upsilon}
\newcommand{\eps}{\varepsilon}

\def \ra {\rightarrow}

\newcommand{\CA}{{\mathcal A}}
\newcommand{\CB}{{\mathcal B}}

\newcommand{\CD}{{\mathcal D}}

\newcommand{\CG}{{\mathcal G}}
\newcommand{\CH}{{\mathcal H}}

\newcommand{\CU}{{\mathcal U}}
\newcommand{\CV}{{\mathcal V}}
\newcommand{\CW}{{\mathcal W}}

\newcommand{\NN}{{\mathbb N}}
\newcommand{\ZZ}{{\mathbb Z}}
\newcommand{\RR}{{\mathbb R}}
\newcommand{\RRp}{{\mathbb R}_{\geq 0}}
\newcommand{\RRm}{{\mathbb R}_{< 0}}
\newcommand{\CCC}{{\mathbb C}}

\newcommand{\TT}{{\mathbb T}}

\newcommand{\SA}{{\mathscr A}}
\newcommand{\SAl}{{\mathscr A}^{l}}
\newcommand{\SAal}{{\mathscr A}^{al}}
\newcommand{\SAql}{{\mathscr A}^{ql}}

\newcommand{\SUl}{{\mathscr U}^l}
\newcommand{\SUal}{{\mathscr U}^{al}}
\newcommand{\SUql}{{\mathscr U}^{ql}}

\newcommand{\HilbV}{{V}}

\newcommand{\UnitaryOper}{{U}}
\newcommand{\BoundedOper}{{B}}
\newcommand{\Ball}{{\mathsf B}}
\newcommand{\bpi}{{\bm{\pi}}}

\newcommand{\mfkDal}{{\mathfrak{D}}^{al}}
\newcommand{\mfkDl}{{\mathfrak{D}}^{l}}

\newcommand{\mfkdal}{{{\mathfrak d}^{al}}}
\newcommand{\mfkdl}{{{\mathfrak d}^{l}}}
\newcommand{\mfkdalsn}{{{\mathfrak d}^{al}_{<0}}}
\newcommand{\mfkdalsp}{{{\mathfrak d}^{al}_{\geq 0}}}

\newcommand{\mfkDalsp}{{\mathfrak{D}}^{al}_{\geq 0}}
\newcommand{\mfkDalsm}{{\mathfrak{D}}^{al}_{< 0}}
\newcommand{\mfkDalp}{{\mathfrak{D}}^{al}_{+}}
\newcommand{\mfkDalm}{{\mathfrak{D}}^{al}_{-}}

\newcommand{\mfkg}{{\mathfrak g}}

\newcommand{\LPA}{{\CG^{\text{lp}}}}

\newcommand{\LPAp}{{\CG^{\text{lp}}_+}}
\newcommand{\LPAm}{{\CG^{\text{lp}}_-}}
\newcommand{\LPApm}{{\CG^{\text{lp}}_\pm}}
\newcommand{\LPAz}{{\CG^{\text{lp}}_0}}
\newcommand{\LPAsp}{{\CG^{\text{lp}}_{\geq 0}}}
\newcommand{\LPAsm}{{\CG^{\text{lp}}_{<0}}}

\newcommand{\alLPA}{{\CG^{\text{al}}}}

\newcommand{\alLPAp}{{\CG^{\text{al}}_+}}
\newcommand{\alLPAm}{{\CG^{\text{al}}_-}}
\newcommand{\alLPAz}{{\CG^{\text{al}}_0}}
\newcommand{\alLPAsp}{{\CG^{\text{al}}_{\geq 0}}}
\newcommand{\alLPAsm}{{\CG^{\text{al}}_{<0}}}

\newcommand{\chf}{{\mathsf f}}
\newcommand{\chg}{{\mathsf g}}
\newcommand{\chh}{{\mathsf h}}

\newcommand{\chq}{{\mathsf q}}

\newcommand{\derF}{{\mathsf{F}}}
\newcommand{\derG}{{\mathsf{G}}}
\newcommand{\derH}{{\mathsf{H}}}
\newcommand{\derQ}{{\mathsf{Q}}}

\newcommand{\OL}{\mathcal{O}(L^{-\infty})}
\newcommand{\Or}{\mathcal{O}(r^{-\infty})}

\newcommand{\Br}{{\mathbb B}}

\newcommand{\Frechet}{{Fr\'{e}chet}}

\def \l {\left(}
\def \r {\right)}

\def \lal {\langle}
\def \ral {\rangle}

\title{Anomalous symmetries of quantum spin chains and a generalization of the Lieb-Schultz-Mattis theorem}

\author{ Anton Kapustin \\
{\it California Institute of Technology, Pasadena, CA 91125}
\and 
Nikita Sopenko \\
{\it Institute for Advanced Study, Princeton, NJ 08540}
}

\date{\today}

\begin{document}

\maketitle

\begin{abstract}
For any locality-preserving action of a group $G$ on a quantum spin chain one can define an anomaly index taking values in the group cohomology of $G$. The anomaly index is a kinematic quantity, it does not depend on the Hamiltonian. We prove that a nonzero anomaly index prohibits any $G$-invariant Hamiltonian from having $G$-invariant gapped ground states. Lieb-Schultz-Mattis-type theorems are a special case of this result when $G$ involves translations. In the case when the symmetry group $G$ is a Lie group, we define an anomaly index which takes values in the differentiable group cohomology as defined by J.-L. Brylinski and prove a similar result.
\end{abstract}

\section{Introduction}

The most basic questions about the dynamics of quantum spin systems concern the properties of their ground states and the low-energy spectrum in the thermodynamic limit. In particular, one would like to know whether there is a gap in the energy spectrum separating ground states from excited states, whether the ground state is unique, and, if a symmetry is present, whether it is spontaneously broken or not.

An early example of a rigorous constraint on the low-energy behavior of quantum spin systems appeared in the works of Lieb, Schultz, and Mattis \cite{LSM} and Affleck and Lieb \cite{AL} in the analysis of the one-dimensional  antiferromagnetic Heisenberg model with a half-integral spin. These authors proved that certain spin chains (i.e. one-dimensional quantum spin systems) cannot have a unique gapped ground state in the thermodynamic limit. It was later realized that this constraint is robust against deformations of the Hamiltonian as long as they preserve the symmetry of the system. In particular, it was conjectured in \cite{chen2011classification} and proved in \cite{ogata2019lieb,ogatatachikawatasaki} that any quantum spin chain with a finite-range Hamiltonian which is invariant under translation symmetry and a projective on-site action of an internal finite symmetry group cannot have a symmetric gapped ground state in the thermodynamic limit. Various generalizations, including higher-dimensional ones, have been proposed over the years and are often referred to as LSM-type theorems due to the work \cite{LSM}. See \cite{tasaki2022lieb} for a review.

Recently, thanks to the developments in the classification of topological phases of matter, it was argued that the underlying reason for (at least some) LSM-type theorems is the anomalous realization of symmetries \cite{Barkeshlietal,ChoHsiehRyu,ChengSeiberg}. This proposal builds on an earlier observation that the robust character of gapless symmetric edges of Symmetry Protected Topological (SPT) phases can often be explained by an anomalous realization of symmetries in the Effective Field Theory of the edge modes \cite{ChenLiuWen,levin2012braiding,SPTanomalies1,SPTanomalies2}. 

The constraints on the low-energy behaviour of quantum field theories (QFTs) due to symmetries also have a long history. It has been  known since the works of Adler \cite{Adler} and Bell and Jackiw \cite{BellJackiw} that in theories with fermions there may be obstructions to gauging continuous global symmetries. It has been argued by 't Hooft \cite{tHooft} that these obstructions are preserved under RG flow and thus a QFT with a non-trivial obstruction for gauging a symmetry $G$ cannot have a gapped vacuum. One can compute these obstructions, often referred to as 't Hooft anomalies, if one has a good control over the QFT dynamics at short distance scales. The task is simplified by the fact that t' Hooft anomalies are robust: they do not change under reasonable deformations of the dynamics. The concept of 't Hooft anomalies has been generalized to discrete symmetries and to space-time diffeomorphism symmetries ('t Hooft anomalies of diffeomorphism symmetries are known as gravitational anomalies). When a discrete symmetry of a QFT  is afflicted with 't Hooft anomalies, a gapped vacuum is not ruled out, but the invariance of 't Hooft anomaly under RG flow means that the low-energy Effective Field Theory cannot be completely trivial. 

Since 't Hooft anomalies of QFTs can be computed from the knowledge of dynamics at arbitrarily short distances, it is natural to expect that one can read them off a lattice regularization of the theory. A concrete way to detect anomalies of internal unitary symmetries of low-dimensional lattice systems was proposed by Nayak and Else  \cite{else2014classifying}. These authors assumed that the symmetry group $G$ acts on a lattice system by (unitary finite-depth) circuits (a special type of automorphisms which map local observables to local observables, see Section 2 for a precise definition). In the one-dimensional case, for any homomorphism from $G$ to the group of circuits, they constructed an element of the degree-3 group cohomology $H^3(G,U(1))$ which we call the anomaly index. The anomaly index is an obstruction to localizing the action of $G$ to a half-line and it trivially vanishes for on-site actions. Nayak and Else argued that if the anomaly index is nonzero then a state obtained by acting on an unentangled state with a circuit cannot be $G$-invariant. While this result is similar to an LSM-type theorem, it is distinct in a couple of ways: it does not involve translation symmetry and it does not immediately lead to constraints on gapped states. More recently, other procedures for defining an anomaly index that allow  one to incorporate anti-unitary symmetries \cite{kawagoe2021anomalies} and translations \cite{seifnashri2023lieb} have been proposed.

In this paper, we perform a general analysis of the action of a symmetry group $G$ on a quantum spin chain by {\it locality-preserving automorphisms.} This is a very general class of automorphisms studied by Ranard, Walter, and Witteveen \cite{ranard2022converse} which includes circuits, Quantum Cellular Automata (QCAs), as well as finite-time evolutions generated by sufficiently local Hamiltonians. For a given action, we define an anomaly index which takes values in $H^3(G,U(1))$  (for both discrete and continuous groups). We then prove that any Hamiltonian which is invariant under an anomalous symmetry cannot have a gapped $G$-invariant ground state in the thermodynamic limit. This includes Lieb-Schultz-Mattis-type theorems as a special case, where $G$ is a product of translation symmetry and a projective on-site unitary symmetry. The anomaly indices for the translation symmetry alone and for the projective on-site symmetry alone vanish, so in this case a non-trivial anomaly arises from an  interplay between the two (in the QFT context, this is referred to as a "mixed anomaly"). 

The case of a Lie group symmetry needs special attention. While group cohomology and the anomaly index are well-defined in this case too, they are impractical, since groups of cochains of an uncountably infinite group are immense (their  cardinality is strictly larger than the cardinality of the continuum). We argue that for Lie groups, a suitable replacement of the ordinary group cohomology is the differentiable group cohomology defined in \cite{brylinski2000differentiable}. Assuming that the action of $G$ on a spin chain is smooth in a natural sense, we define an anomaly index with values in the differentiable group cohomology. We prove that if it is nonzero, then a Hamiltonian invariant under such a symmetry cannot have a gapped $G$-invariant ground state in the thermodynamic limit. The original Lieb-Schultz-Mattis theorem for $SO(3)$-invariant spin chains is a special case of this result. Differentiable group cohomology appears to be a suitable replacement for ordinary group cohomology in other related problems, such as the classification of SPT phases. 

We emphasize that the anomaly index for lattice spin systems is purely kinematic, i.e. it does not depend on the choice of the Hamiltonian. In contrast, although 't Hooft anomalies in QFT are quite robust, it is not known how to define them without specifying dynamics. Another difference between QFT and lattice systems is which anomalies can be nonzero. In the case of quantum spin chains, there can be a mixed anomaly between translation symmetry and an on-site unitary  symmetry, leading to LSM-type theorems, while in 1+1d QFT mixed anomalies between translations and internal symmetries are impossible. Instead, the QFT interpretation of LSM-type results involves "emanant" internal symmetries of continuum Effective Field Theories which do not exist on the lattice \cite{ChoHsiehRyu,ChengSeiberg}. In the opposite direction, we prove that a compact connected Lie group symmetry of a quantum spin chain cannot have a nonzero anomaly index. This in a marked contrast to QFT, where 1+1d theories of chiral bosons provide some of the most well-known examples of 't Hooft anomalies for compact connected Lie groups.

The paper is organized as follows. In Section \ref{sec:preliminaries}, we review the general structure of locality-preserving automorphisms of one-dimensional lattice spin systems. In Section \ref{sec:anomalyindex}, we define the anomaly index for a given action of a symmetry group $G$. We discuss some well-known examples of anomalous symmetries, both with and without translations, and show how they are compatible with our definition. In Section \ref{sec:LSM}, we prove the main theorem of the paper. In Appendix \ref{app:LieGroup}, we generalize the discussion in the main text to the case of a Lie group symmetry. This is technically more involved because one needs to take into account the smooth manifold structure of the group. In particular, we define the differentiable cohomology of Lie groups where the anomaly indices take values and introduce a class of automorphisms of spin systems which can be used to realize a Lie group symmetry. We also prove that for compact connected Lie groups the anomaly index always vanishes. Appendix \ref{app:cocycle} consists of computations needed to verify that the anomaly index is well-defined.
\\

\noindent
{\bf Acknowledgements:} N.S. would like to thank Nathan Seiberg, Wilbur Shirley and, especially, Sahand Seifnashri for the discussion of anomalies of lattice systems. A.K. is supported by the U.S.\ Department of Energy, Office of Science, Office of High Energy Physics, under Award Number DE-SC0011632, and by the Simons Investigator Award. N.S. is supported by NSF Grant PHY-0503584 and Ambrose Monell Foundation.
\\

\section{Preliminaries} \label{sec:preliminaries}

\subsection{1d quantum spin systems}

For a complex Hilbert space $\HilbV$, we let $\BoundedOper(\HilbV)$ be the normed $*$-algebra of bounded operators on $\HilbV$. We denote the unitary group of $\HilbV$ by $\UnitaryOper(\HilbV)$.

Let $\Lambda$ be a countable set. For a finite-dimensional complex Hilbert space $\HilbV$ and a finite subset $\Gamma \subset \Lambda$, we define $\SA_{\Gamma} := \bigotimes_{j \in \Gamma} \BoundedOper(\HilbV)$. {\it The algebra of local observables} $\SAl$ of a quantum spin  system with an on-site Hilbert space $V$ is defined to be the colimit (of normed $*$-algebras) $\SAl := \varinjlim_{\Gamma} \SA_{\Gamma}$ over finite subsets $\Gamma$ with respect to the canonical embeddings $ \SA_{\Gamma} \to \SA_{\Gamma'}$ for $\Gamma \subset \Gamma'$. This is a normed $*$-algebra. If $\CA \in \SAl$ belongs to $\SA_{\Gamma}$, we say that $\CA$ is {\it strictly localized} or {\it supported} on $\Gamma$. The completion of the algebra $\SAl$ with respect to the operator norm gives the $C^*$-algebra of quasi-local observables $\SAql$. $\SAql$ belongs to the class of Uniformly Hyperfinite (UHF) $C^*$-algebras. It depends on the on-site Hilbert space $V$. The groups of unitary elements in $\SAl$ and $\SAql$ will be denoted $\SUl$ and $\SUql$, respectively.

In the following, we let $\Lambda$ be the set $\ZZ \subset \RR$ embedded into a one-dimensional Euclidean space. For a subset $X \subset \RR$, we let $\SAl_{X}$ be $\SAl_{X \cap \Lambda} := \varinjlim_{\Gamma} \SA_{\Gamma \subset (X \cap \Lambda)}$.  We let $\SAql_X\subset\SAql$ be the norm-completion of the subalgebra $\SAl_X\subset\SAl$. When $X = \RR_{\geq 0}$ (resp. $\RR_{<0}$) we denote $\SAql_{X}$ by $\SAql_{\geq 0}$ (resp. $\SAql_{<0}$). The groups of unitary elements in $\SAl_X$ and $\SAql_X$ will be denoted $\SUl_X$ and $\SUql_X$, respectively.

Given two lattice systems with on-site Hilbert spaces $\HilbV_1$ and $\HilbV_2$ and their algebras of quasilocal observables $\SAql_1$ and $\SAql_2$, we define their stacking as the lattice system with the on-site Hilbert space $\HilbV_1\otimes\HilbV_2$. The corresponding algebra of quasi-local observables is the tensor product $\SAql_1\otimes\SAql_2$. The tensor product is unambiguously defined because UHF algebras are nuclear \cite{bratteli2012operator}. The stacking operation is commutative and associative up to a canonical isomorphism.

\subsection{Locality-Preserving Automorphisms}\label{sec:LPAs}

Symmetries of quantum spin systems are realized by automorphisms\footnote{In this paper, by automorphisms of $*$-algebras we always mean $*$-automorphisms. In the case of $C^*$-algebras, $*$-automorphisms automatically preserve the norm \cite{bratteli2012operator}.} of the algebra of observables. Usually, they are required to preserve some notion of locality that depends on the metric on the lattice. One way to characterize the locality of an automorphism $\alpha:\SAql\ra\SAql$ is by how hard it is to approximate $\alpha(\CA)$, $\CA\in\SAl$ in the norm by local observables with a slightly bigger support than $\CA$. More precisely, for any $r>0$ we may consider a non-negative but possibly infinite quantity
\beq \label{eq:autlocality}
\sup_{\Gamma} \sup_{\CA \in \SAql_{\Gamma}} \inf_{\CB \in \SAql_{\Ball_{\Gamma}(r)}} (\|\alpha(\CA) - \CB\|/\|\CA\|),
\eeq
where the supremum is taken over all closed intervals $\Gamma \subset \RR$ and $\Ball_{\Gamma}(r)$ is the set of points $x \in \RR$ such that $|x-y| \leq r$ for some $y \in \Gamma$. We say that $\alpha$ is locality-preserving\footnote{Note that this definition of locality excludes reflection symmetry and time-reversal symmetry of a 1d spin chain.} if there exists a function $f:\RR_{\geq 0} \to \RR_{\geq 0}$ with $\lim_{r \to \infty} f(r) = 0$ that upper-bounds the expression  (\ref{eq:autlocality}). In this case, we also say that $\alpha$ has $f(r)$-tails.

If $\alpha$ is such that $f(r)$ can be chosen to vanish for sufficiently large $r$, then there exists $R>0$ such that $\alpha(\CA) \subset \SA_{\Ball_X(R)}$ for any $X\subset\RR$ and any $\CA \in \SA_X$. Such automorphisms are called {\it Quantum Cellular Automata} (QCAs). They form a group with respect to composition.
An example of a QCA is an automorphism of the form
\beq
\prod_{p=-\infty}^\infty \Ad_{\CU_p},
\eeq
where $\Ad_{\CU}$ is an automorphism $\Ad_{\CU}(\CA):=\CU \CA \CU^*$ and local unitary observables $\CU_p\in \SUl$, $p\in\ZZ$, are supported on non-intersecting intervals with  uniformly bounded diameters. Following \cite{ranard2022converse}, we call such automorphisms {\it block-partitioned unitaries}. A composition of a finite number of block-partitioned unitaries is called a (finite-depth unitary) {\it circuit}. Circuits form a group with respect to composition. Every  circuit is a QCA, but not every QCA is a circuit. For example, a  QCA which shifts the whole spin system by one site to the right (i.e. an automorphism $\tau$ of $\SAl$ uniquely defined by the condition that for all $j\in\ZZ$ it maps $\SA_j\simeq \BoundedOper(\HilbV)$ to $\SA_{j+1}\simeq \BoundedOper(\HilbV)$ in the obvious manner) is not a circuit \cite{GNVW}.

The problem of classifying QCAs modulo circuits has been posed and solved in \cite{GNVW}. One can define a homomorphism from the group of QCAs to $\ZZ[\{\log p_i\}_{i \in J}]$, where $\{p_i\}_{i\in J}$ is the set of all positive primes dividing $\dim \CV$. The group of circuits is precisely the kernel of this homomorphism. We call this homomorphism the {\it GNVW index}. It is surjective, and QCAs whose GNVW indices generate $\ZZ[\{\log p_i\}_{i \in J}]$ can be constructed as follows. Let us fix an isomorphism $V\simeq \left(\bigotimes_{i\in J} V_{p_i}\right)\otimes V'$, where $V_{p_i}$ is a Hilbert space of dimension $p_i$ and $V'$ is a Hilbert space of dimension $\dim V/\prod_{i\in J} p_i$. Then $\SAql$ is identified with a tensor product of $|J|+1$ UHF algebras, and for any $i\in J$ there is a well-defined QCA which shifts the sub-algebra corresponding to $i$ by one site to the right and acts as the identity automorphism of the other factors. We will call such an automorphism a  generalized translation. Generalized translations (for a fixed isomorphism 
$V\simeq \left(\bigotimes_{i\in J} V_{p_i}\right)\otimes V'$) form an abelian group isomorphic  to $\ZZ[\{\log p_i\}_{i \in J}]$.

\begin{remark}\label{rem:semidirect} 
Surjectivity of the GNVW index means that the group of QCAs is an extension of $\ZZ[\{\log p_i\}_{i \in J}]$ by the group of circuits. In fact, the above construction of QCAs generating $\ZZ[\{\log p_i\}_{i \in J}]$ is a homomorphism from $\ZZ[\{\log p_i\}_{i \in J}]$ to the group of QCAs and thus exhibits the latter as a semi-direct product of the group of circuits and $\ZZ[\{\log p_i\}_{i \in J}]$. This identification depends on the choice of the isomorphism $V\simeq \left(\bigotimes_{i\in J} V_{p_i}\right)\otimes V'$.
\end{remark}

The GNVW index is additive under the stacking of systems. For example, although the shift automorphism $\tau$ is not a circuit, the automorphism $\tau\otimes\tau^{-1}$ acting on two copies of the same system has a vanishing GNVW index and thus is a circuit. 

Ref. \cite{ranard2022converse} explored the question about when a Locality-Preserving Automorphism (LPA) \footnote{Note that in \cite{ranard2022converse} Locality-Preserving Automorphisms were called Approximately Locality-Preserving Unitaries (ALPU).} can be described by a local Hamiltonian evolution with a time-dependent Hamiltonian. This question is analogous to the question of which QCAs are circuits. Indeed, a block-partitioned unitary can be thought of as an evolution automorphism induced by a (rather special) time-independent finite-range Hamiltonian, while a circuit is induced by a piecewise-constant time-dependent 
finite-range Hamiltonian. The authors of Ref.\cite{ranard2022converse} proved that an obstruction to realizing an LPA by Hamiltonian evolution arises from a generalized GNVW index which again takes values in $\ZZ[\{\log p_i\}_{i \in J}]$ (Theorem 5.8 in \cite{ranard2022converse}). Moreover, they showed  that this index is in some sense complete.

LPAs form a group. Indeed, according to Lemma 3.8.iv in \cite{ranard2022converse}, the inverse of an automorphism with $f(r)$-tails has $4 f(r)$-tails. It is also easy to see that if $\alpha$ has $f(r)$-tails and $\beta$ has $g(r)$-tails, then $\alpha\circ\beta$ has $h(r)$-tails, where $h(r)=f(r/2)+g(r/2)+f(r/2) g(r/2)$. We denote the group of LPAs by $\LPA$. We also introduce its various subgroups. We let $\LPAsp$, (resp. $\LPAsm$ ) be the subgroup of LPAs which preserve $\SAql_{\geq 0}$ (resp. $\SAql_{< 0}$) and act trivially on $\SAql_{< 0}$ (resp. $\SAql_{\geq 0}$). We let $\LPAz \subset \LPA$ be the subgroup of automorphisms of the form $\Ad_{\CU}$ for a unitary $\CU \in \SAql$. Note that $\LPAsp$ and $\LPAsm$ act on $\LPAz$ by conjugation. We denote the subsets $\LPAz \LPAsp$ and $\LPAz \LPAsm$ by $\LPAp$ and $\LPAm$, respectively. $\LPAp$ and $\LPAm$ are subgroups of $\LPA$.

The following property will be crucial in what follows.
\begin{lemma} \label{lma:alphadecomposition}
A locality-preserving automorphism $\alpha$ has a trivial GNVW index if and only if it admits a decomposition $\alpha = \alpha_{<0} \alpha_0 \alpha_{\geq 0}$ for some $\alpha_{<0} \in \LPAsm$, $\alpha_{0} \in \LPAz$, $\alpha_{\geq 0} \in \LPAsp$.
\end{lemma}

\begin{proof}
By Proposition 5.14 from \cite{ranard2022converse}, if $\alpha$ admits such a decomposition, then its GNVW index is trivial. The converse follows from Theorem 5.9 of  \cite{ranard2022converse}. Let $\{v^{(k,a)}_n\}_{n \in \ZZ, k \in \NN, a \in \{1,2\}} \in \SUl$ be the set of unitary observables constructed in this proof. We define the sets $S_+^{(k,a)} \subset \ZZ$ recursively as follows. We let $S^{(1,1)}_+$ be the set of $n \in \ZZ$, such that the support of $v^{(1,1)}_n$ has no overlap with $\RR_{<0}$. For $k \geq 1$, we let $S^{(k,2)}_+$ be the set of $n \in \ZZ$, such that the support of $v^{(k,2)}_n$ has no overlap with the left-most point of the supports of $v^{(k,1)}_m$ for $m \in S^{(k,1)}_+$. For $k > 1$, we let $S^{(k,1)}_+$ be the set of $n \in \ZZ$, such that the support of $v^{(k,1)}_n$ has no overlap with the left-most point of the supports of $v^{(k-1,2)}_m$ for $m \in S^{(k-1,2)}_+$. The sets $S_-^{(k,a)} \subset \ZZ$ are defined similarly, with $\RR_{<0}$ replaced by $\RR_{\geq 0}$ and the left-most replaced by the right-most. We let $S_0^{(k,a)} \subset \ZZ$ be the sets which are complements of $S_+^{(k,a)} \cup S_-^{(k,a)}$ in $\ZZ$. Let $\CV^{(k)} = \l \prod_{n \in S_0^{(k,2)}} v^{(k,2)*}_n \r \l \prod_{n \in S_0^{(k,1)}} v^{(k,1)*}_n \r$, and let $\CU^{(k)} = \CV^{(k)} ... \CV^{(2)} \CV^{(1)}$. By the bound on $\|v^{k,a)}_n -1 \|$, the unitaries $\CU^{(k)}$ converge to some $\CU \in \SUql$. Clearly, we get the decomposition $\alpha = \alpha_{<0} \alpha_{\geq 0} \Ad_{\CU}$ for $\alpha_{<0} \in \LPAsm$, $\alpha_{\geq 0} \in \LPAsp$. We can take $\alpha_0$ to be $\Ad_{\alpha_{\geq 0}(\CU)}$.
\end{proof}

\begin{remark}\label{rem:semidirectLPA}
As in Remark \ref{rem:semidirect}, by choosing an isomorphism $V\simeq \left(\bigotimes_{i\in J} V_{p_i}\right)\otimes V'$, we get a homomorphism from $\ZZ[\{\log p_i\}_{i \in J}]$ to the group of LPAs. Thus the group of LPAs is exhibited as a semidirect product of the group of LPAs with a vanishing GNVW index and $\ZZ[\{\log p_i\}_{i \in J}]$.
\end{remark}

The decomposition whose existence is ensured by Lemma \ref{lma:alphadecomposition} is not unique, and the next lemma characterizes this non-uniqueness.

\begin{lemma}\label{lma:nonuniqueness}
Suppose $\alpha\in\LPA$ has a trivial GNVW index. For any two decompositions $\alpha = \alpha_{<0} \alpha_0 \alpha_{\geq 0}={\tilde\alpha}_{<0} {\tilde\alpha}_0 {\tilde\alpha}_{\geq 0}$ as in Lemma \ref{lma:alphadecomposition}, we must have
$\tilde\alpha_{<0}\alpha_{<0}^{-1}\in\LPAz$ and $\tilde\alpha_{\geq 0}\alpha_{\geq 0}^{-1}\in\LPAz$.
\end{lemma}
\begin{proof}
Let $\beta_{<0}=\tilde\alpha_{<0}\alpha_{<0}^{-1}$ and $\beta_{\geq 0}=\tilde\alpha_{\geq 0}\alpha_{\geq 0}^{-1}$. Then $\beta_{<0} \in \LPAsm$, $\beta_{\geq 0} \in \LPAsp$, and $\beta_{<0} \beta_{\geq 0} = \Ad_{\CU}$ for some $\CU \in \SUql$.
Thus $\Ad_\CU$ induces an automorphism of $\SAql_{<0}$ as well as an automorphism of $\SAql_{\geq 0}$. On the other hand, $\Ad_\CU$ is asymptotically identity on $\SAql$ and therefore on $\SAql_{<0}$ and $\SAql_{\geq 0}$ (i.e. for any $\eps>0$ there is an $R>0$ such that 
$\|(\left(\Ad_\CU-\Id\right)\vert_{\SAql_{(-\infty,-R]}}\|\leq \eps$ and $\|(\left(\Ad_\CU-\Id\right)\vert_{\SAql_{[R,+\infty)}}\|\leq \eps$). Therefore by Lemma 3.1 from \cite{Lance}, $\CU=\CU_{<0}\CU_{\geq 0}$ for some $\CU_{<0}\in\SUql_{<0}$ and $\CU_{\geq 0}\in\SUql_{\geq 0}$. Then $\beta_{<0}=\Ad_{\CU_{<0}}$ and $\beta_{\geq 0}=\Ad_{\CU_{\geq 0}}$.
\end{proof}

\begin{corollary}\label{cor:intersection}
The intersection $\LPAp\cap\LPAm$ coincides with $\LPAz$.
\end{corollary}
\begin{proof}
Clearly, $\LPAz\subset \LPAp\cap\LPAm$. To show the opposite inclusion, suppose $\beta=\Ad_\CU \alpha_{\geq 0}={\tilde \alpha}_{<0} \Ad_{\tilde\CU}$ for some $\CU,\tilde\CU\in\SUql$. Then by Lemma \ref{lma:nonuniqueness} we have $\alpha_{\geq 0}\in\LPAz$ and ${\tilde\alpha}_{<0}\in\LPAz $, and thus $\beta\in\LPAz$.
\end{proof}

\begin{remark}\label{rem:normalsubgroups}
Lemma \ref{lma:alphadecomposition} implies that $\LPAp$ and $\LPAm$ are normal subgroups of the group of LPAs with a trivial GNVW index. It is also not difficult to show that conjugation with a generalized translation $\tau$ maps $\LPApm$ to itself. Indeed, by Theorem 5.15 \cite{ranard2022converse}, the automorphism $\tau \alpha \tau^{-1}$ for $\alpha \in \LPAsp$ has a trivial GNVW index. Therefore, by Lemma \ref{lma:alphadecomposition}, it admits a decomposition $\tau \alpha \tau^{-1} = \alpha_- \alpha_0 \alpha_+$. Lemma 3.1 from \cite{Lance} implies that $\alpha_- \in \LPAz$. Therefore both $\LPAp$ and $\LPAm$ are normal subgroups of $\LPA$.
\end{remark}

\subsection{Group cohomology}

Let $G$ be a group. For $n \geq 1$ and $k=0,...,n$, we define the maps $d_k:G^n \to G^{n-1}$ via
\beq
d_k (g_1,...,g_n) = 
\begin{cases}
    (g_2,...,g_n)\ \ \text{for}\ k=0,\\
    (g_1,...,g_k g_{k+1},...,g_n)\ \ \text{for}\ 0<k<n,\\
    (g_1,...,g_{n-1})\ \ \text{for}\ k=n.
\end{cases}
\eeq

For an abelian group $A$, we let $(C^{\bullet}(G,A),d)$ be the cochain complex with $C^{n}(G,A)$ being the abelian group of maps $f:G^n \to A$ and with the differential being defined by $d f := d_0^* f - d_1^*f + ... + (-1)^n d^*_n f$. The group cohomology $H^{\bullet}(G,A)$ of $G$ with coefficients in $A$ can be computed as the cohomology of this cochain complex.

For example, a 2-cocycle is a function  $\lambda:G\times G\ra A$ satisfying 
\beq
\lambda(g_2,g_3)+\lambda(g_1,g_2 g_3)-\lambda(g_1 g_2, g_3)-\lambda(g_1,g_2)=0,\quad\forall g_1,g_2,g_3\in G,
\eeq
and a 2-coboundary is a function of the form  $\lambda(g_1,g_2)=\phi(g_1)+\phi(g_2)-\phi(g_1 g_2)$ for some $\phi:G\ra A$. Elements of $H^2(G,A)$ label equivalence classes of central extensions of $G$ by $A$. A 3-cocycle is a function $\omega:G\times G\times G\ra A$ satisfying
\beq
\omega(g_2,g_3,g_4)+\omega(g_1,g_2g_3,g_4)+\omega(g_1,g_2,g_3)-\omega(g_1g_2,g_3,g_4)-\omega(g_1,g_2,g_3g_4)=0,\quad \forall g_1,...,g_4\in G,
\eeq
while a 3-coboundary is a function of the form
\beq
\omega(g_1,g_2,g_3)=\psi(g_2,g_3)+\psi(g_1,g_2g_3)-\psi(g_1g_2,g_3)-\psi(g_1,g_2),
\eeq 
for some $\psi:G\times G\ra A$.

Equivalently, one may define group cohomology of $G$ with coefficients in $A$ as the singular cohomology $H^n_{sing}(BG,A)$ of the classifying space $BG$ with coefficients in $A$. The classifying space $BG$ is a topological space whose homotopy type is uniquely determined by requiring that  $\pi_1(BG)=G$ and all other homotopy groups are  trivial. Such a space exists for every group $G$ \cite{Hatcher}. For example, one can take $B\ZZ=S^1$, and thus for $n>1$, $H^n(\ZZ,A)=H^n_{sing}(S^1,A)=0$, while  $H^1(\ZZ,A)=A$. 

For each $n\geq 0$ and each $A$ the assignment $G\mapsto H^n(G,A)$ defines a contravariant functor from the category of groups to the category of abelian groups. Thus for every homomorphism of groups $\phi:G'\ra G$ there is a pull-back homomorphism $\phi^*:H^n(G,A)\ra H^n(G',A)$ . On the level of cocycles, it is defined by pre-composing a cocycle of $G$ with $\phi\times \ldots\times \phi$.

Group cohomology is most useful for finite or countably infinite groups. For uncountably infinite groups it can be very large and hard to compute, see \cite{Milnor} and references therein. For Lie groups, whether compact or non-compact,  differentiable group cohomology defined in \cite{brylinski2000differentiable} and reviewed in Appendix A is better behaved and often more useful.

\section{Anomaly index for symmetries of quantum spin chains} \label{sec:anomalyindex}

In this section, we analyze the case of an abstract  group $G$ which acts on a 1d spin system by locality-preserving automorphisms. For any  homomorphism $\alpha:G \to \LPA$, we define an $H^3(G,U(1))$-valued index that we refer to as an {\it anomaly index}. In case when the image of $\alpha$ lands in the subgroup of LPAs with a trivial GNVW index, the definition is essentially the same as in \cite{else2014classifying}. In this special case the anomaly index is an obstruction to represent $\alpha$ as $\beta_-  \beta_+$ for some commuting 
homomorphisms $\beta_-:G \to \LPAm$ and $\beta_+:G\ra \LPAp$. One can interpret it as an obstruction to promote the symmetry to a spatially-localized symmetry  or a gauge symmetry. In particular, it gives an obstruction to define the action of the symmetry on a system with a boundary so that far away from the boundary it acts in the same way as $\alpha$.

When $G$ is a Lie group, it is more natural to consider homomorphisms which are smooth in some sense. The anomaly index can be defined in this situation as well, but this case is more technical and is discussed in Appendix \ref{app:LieGroup}.

\subsection{Anomalous symmetries in quantum mechanics}

Before discussing the case of spin chains, let us review symmetry actions in quantum mechanics. If $\CH$ is the Hilbert space of a quantum mechanical systems, then an action of a group $G$ on $\CH$ is a homomorphism $\alpha: G \to PU(\CH)={\rm Aut}(B(\CH))$ from $G$ to the projective unitary group $PU(\CH)$ of the Hilbert space $\CH$. It is always possible to lift it to a homomorphism $\tilde{G} \to U(\CH)$ from $\tilde{G}$ to the unitary group $U(\CH)$ of $\CH$, where $\tilde{G}$ is a central extension of $G$ by $U(1)$, but it might not be possible to lift it to a homomorphism $G \to U(\CH)$. The obstruction is given by an element of $H^2(G, U(1))$. 

To construct a cocycle for this class, we can choose a map $\CV:G \to U(\CH)$ (which is not necessarily a homomorphism) that lifts $\alpha: G \to PU(\CH)$. Since $(d^*_2 \alpha)(d^*_0 \alpha)(d^*_1 \alpha)^{-1} = \Id$ on $G^2$, the unitary $(d^*_2 \CV)(d^*_0 \CV)(d^*_1 \CV)^{-1}$ is proportional to the identity and therefore defines an element of $ C^2(G,U(1))$. It is easy to see that it is a 2-cocycle and therefore defines a class in $H^2(G,U(1))$. Moreover, the class is independent of the choice of $\CV$. Indeed, two different choices of $\CV$ are related by a function $G \to U(1)$. Therefore changing $\CV$ can affect the cocycle $(d^*_2 \CV)(d^*_0 \CV)(d^*_1 \CV)^{-1}$ at most by a coboundary.

A non-trivial class in $H^2(G,U(1))$ prohibits the existence of $G$-invariant pure states on the $C^*$-algebra $\CB(\CH)$. In particular, a $G$-invariant Hamiltonian on $\CH$ cannot have a $G$-invariant pure ground state. One can also regard a non-trivial class in $H^2(G,U(1))$ as an obstruction to gauging $G$, since it prohibits one from defining a non-trivial subspace of $G$-invariant vectors in $\CH$.

\subsection{Anomalous symmetries of quantum spin chains}

Let us now define an $H^3(G,U(1))$-valued index for a homomorphism $\alpha:G \to \LPA$.

First, we consider the case when all elements in the image of $\alpha$ have a trivial GNVW index. By decomposing $\alpha$ as in Lemma \ref{lma:alphadecomposition}, we can choose a map $\beta:G \to \LPAp$ (which is not necessarily a homomorphism) such that the image of $\alpha  \beta^{-1}$ belongs to $\LPAm$. Since $(d^*_2 \alpha)  (d^*_0 \alpha) (d^*_1 \alpha)^{-1}=\Id$, using Remark \ref{rem:normalsubgroups} we get $(d^*_2 \beta) (d^*_0 \beta) (d^*_1 \beta)^{-1}\in\LPAm$. Then by Corollary \ref{cor:intersection}, there exists a map $\CV:G^2 \to \SUql$, such that $(d^*_2 \beta) (d^*_0 \beta)  (d^*_1 \beta)^{-1} = \Ad_{\CV}$. The expression 
\beq\label{eq:omega}
\omega = (d^*_3 \CV) (d^*_1 \CV) (d^*_2 \CV)^{-1} ((d^*_3 d^*_2 \beta) (d^*_0 \CV))^{-1}
\eeq
is a multiple of the identity and therefore defines an element of $C^3(G,U(1))$. In Appendix \ref{app:CocycleConditionD}, we show that it satisfies the cocycle condition and therefore defines a class $[\omega] \in H^3(G,U(1)) \cong H^4_{sing}(BG,\ZZ)$. (The computation is identical to that in Appendix B of \cite{else2014classifying}). 

\begin{prop}
The class $[\omega] \in H^3(G,U(1))$ does not depend on the choice of $(\beta, \CV)$.
\end{prop}

\begin{proof}
For a fixed $\beta$, any two choices of $\CV$ differ by some 
$U(1)$-valued functions on $G^2$. Such functions manifestly modify $\omega$ by a coboundary. Therefore, the class $[\omega]$ can depend only on the choice of $\beta$.

Let us make a choice of $(\beta, \CV)$. By Lemma  \ref{lma:nonuniqueness}, any other choice of $\beta$ is given by $\tilde{\beta} = \beta  \Ad_{\CU}$ for some function $\CU:G \to \SUql$. We can take the corresponding $\CV$ to be 
\beq
\tilde{\CV} = (d^*_2\beta(\CU)) \CV (d^*_{1} \beta)((d^*_{0} \CU) (d^*_{1} \CU)^{-1}),
\eeq
One can argue that the class $[\omega]$ is independent of $\CU$ by continuity, as the former is discrete, while the latter can be continuously deformed to the identity thanks to the connectedness of the group $\SUql$. Instead, in Appendix \ref{app:IndependenceDiscrete} we show that the cocycle $\omega$ is the same for $(\beta,\CV)$ and $(\tilde{\beta}, \tilde{\CV})$ by a direct computation.\footnote{The same computation can be found in Appendix B of \cite{else2014classifying}.} Thus, the class $[\omega]$ is independent of the choice of $\beta$.

\end{proof}

\begin{remark}
Instead of restrictions to the right half-chain, one can use restrictions to the the left one and similarly define a class in $H^3(G,U(1))$. One can show that this class is the negative of the class defined above. We omit the details. 
\end{remark}

\begin{remark}\label{rem:onsiteactions}
If $G$ acts on-site (i.e. for all $g\in G$ and all $j\in\ZZ$, $\alpha(g)$ maps $\SAl_j$ to itself), then one can choose $\beta: G\ra\LPAp$ to be a homomorphism. Hence, $\CV(g,h)=1$ for all $g,h\in G$. Thus for on-site actions the anomaly index is trivial. 
\end{remark}

Consider now the general case when the image of $\alpha$ does not necessarily consist of automorphisms with a vanishing GNVW index. Let us fix an isomorphism $\HilbV \to \left(\bigotimes_{i \in J} \HilbV_{p_i}\right)\otimes \HilbV'$, where $\{p_i\}_{i\in J}$ is the set of all primes dividing  $\dim \HilbV$, $\HilbV_{p_i}$ is a $p_i$-dimensional Hilbert space, and $\HilbV'$ is a Hilbert space of dimension $\dim \HilbV / \prod_{i \in J} p_i$. By Remark \ref{rem:semidirectLPA}, $\LPA$ can be identified with a semi-direct product of the group of LPAs with a vanishing GNVW index and the group of generalized translations. Thus $\alpha:G\ra \LPA$ induces a homomorphism $\tau$ from $G$ to the abelian group of generalized translations. Consider a second copy of our system on which the symmetry element $g \in G$ acts by $\tau(g)^{-1}$. Stacking this system with the original system, we get a composite system on which $G$ acts via the automorphisms $\alpha(g) \otimes \tau(g)^{-1}$ each of which has a trivial GNVW index. We define the anomaly index of $\alpha$ to be the anomaly index of $\alpha\otimes\tau^{-1}:G\ra\LPA$ acting on the composite system.

To check that the anomaly index does not depend on the choice of $\tau$, consider a different choice of an isomorphism $\HilbV \to \left(\bigotimes_{i \in J} \HilbV_{p_i}\right)\otimes V'$ which gives a different homomorphism $\tau':G\ra \ZZ\left[\{\log p_i\}_{i\in J}\right]$. Let us consider the action of $G$ on four copies of the original system by $\alpha \otimes \tau^{-1} \otimes \tau'^{-1} \otimes \tau$. Since the anomaly index is additive under stacking, the anomaly indices of $\alpha\otimes\tau^{-1}$ and $\alpha\otimes \tau'^{-1}$ differ by the anomaly index of $\tau \otimes\tau'^{-1}$ minus the anomaly index of $\tau \otimes \tau^{-1}$. Since the circuits $\tau \otimes\tau'^{-1}$ and $\tau \otimes \tau^{-1}$ differ by a conjugation with a circuit (in fact, by a block-partitioned unitary), their anomaly indices are the same.

\begin{remark}\label{rem:functoriality}
Suppose $G$ acts on a 1d lattice system via a homomorphism $\alpha:G\ra\LPA$. Let $G'$ be some other group and $\phi:G'\ra G$ be a homomorphism, so that $G'$ acts on the same lattice system via $\alpha'=\alpha \phi$. Then it is clear that the anomaly index of $\alpha'$ is the pull-back of the anomaly index of $\alpha$ via $\phi$. We refer to this property of the anomaly index as functoriality. 
\end{remark}

\begin{example}
\label{example:z2}
Let $G = \ZZ/2$. We identify $\ZZ/2$ with the multiplicative group of $\{1,-1\}$. We have $H^3(\ZZ/2,U(1)) = \ZZ/2$. Let us describe an action of $\ZZ/2$ by circuits which has a nonzero anomaly index \cite{ChenLiuWen,levin2012braiding}. Consider a spin system with a two-dimensional on-site Hilbert space $\HilbV$. Let $X_j$, $Y_j$, $Z_j$ be observables corresponding to the Pauli matrices $\sigma^x$, $\sigma^y$, $\sigma^z$ acting on the site $j$, and let $\gamma$ be an automorphism of $\SAql$ uniquely defined by 
\begin{equation}\label{eq:Z2auto}
    \gamma(Z_j) = - Z_j,\quad \gamma(X_j) = Z_{j-1} X_j Z_{j+1}.
\end{equation}
It corresponds to a conjugation with a formal infinite product
\beq
\prod_{j \in \ZZ}\l e^{\frac{\pi i}{4} Z_j Z_{j+1}}\r \prod_{j \in \ZZ} X_j.
\eeq
Since $\gamma \gamma = \Id$, we get an action of $\ZZ/2$ on the spin system by letting $\alpha(-1)=\gamma,\alpha(1)=\Id$.  

Let us show that this symmetry is anomalous. We can choose a restriction to the right half-chain to be a conjugation with 
\beq
\prod_{j \in \ZZ_{\geq 0}}\l e^{\frac{\pi i}{4} Z_j Z_j+1{}}\r \prod_{j \in \ZZ_{\geq 0}} X_j.
\eeq
It squares to $\Ad_{Z_0}$. Therefore we can choose  $\CV(-1,-1)$ to be $Z_0$ and $\CV(1,1)=\CV(1,-1)=\CV(-1,1)=1$. The resulting class $[\omega] \in H^3(\ZZ/2,U(1))$ is non-trivial. 

\end{example}

\begin{example}\label{example:LSMfinite}
Let us describe an example where there is a ``mixed anomaly'' between translations and a finite on-site symmetry, even though each of them separately has a vanishing anomaly index.

Let $G_0$ be a finite group which acts on a finite-dimensional Hilbert space $\HilbV$ via a projective representation $\bpi:G_0\ra PU(\HilbV)$. We can regard $\bpi$ as a map $\bpi:G_0\ra U(\HilbV)$ such that
\beq
\bpi(g g')=\rho(g,g')\bpi(g)\bpi(g'),\quad \forall g,g'\in G_0,
\eeq
where $\rho\in C^2(G_0,U(1))$ is a 2-cocycle. We define an action of $G=G_0\times\ZZ$ on the spin chain with the on-site Hilbert space $\HilbV$ as a homomorphism $G\ra \alLPA$ defined by
\beq
(g,n)\mapsto \tau^n\circ \prod_{j\in\ZZ} \Ad_{\bpi_j(g)},\quad (g,n)\in G_0\times\ZZ,
\eeq 
where $\tau$ is a translation to the right by one site and $\bpi_j(g)\in\SUl_j$ is the unitary observable on site $j$ corresponding to $\bpi(g)\in U(\HilbV)$. 

By the universal coefficient formula and K\"unneth formula, $H^3(G,U(1))\simeq H^3(G_0,U(1))\oplus H^2(G_0,U(1))$. The map $H^3(G,U(1))\ra H^3(G_0,U(1))$ is the restriction (pull-back) map. The restriction of the anomaly index to $G_0$ vanishes, since the action of $G_0$ is on-site (see Remark \ref{rem:onsiteactions}). Therefore the anomaly index can be regarded as taking values in $H^2(G_0,U(1))$ and describes a ``mixed anomaly'' between $G_0$ and $\ZZ$. 

The map $H^3(G,U(1))\ra H^2(G_0,U(1))$ is induced by the slant product (see e.g. Chapter 3.B of \cite{Hatcher}) with the generator $[1]$ of the group homology $H_1(\ZZ,\ZZ)\simeq H_1^{sing}(S^1)\simeq\ZZ$. Writing $U(1)$ multiplicatively (i.e. identifying it with the unit circle in the complex plane), it maps a 3-cocycle 
$\omega(g,n; g',n'; g'',n'')\in C^3(G_0\times\ZZ,U(1))$ to a 2-cocycle
\beq\label{slantproduct}
(\omega/[1])(g,g')=\frac{\omega(e,1;g,0;g',0)\omega(g,0;g',0;e,1)}{\omega(g,0;e,1;g',0)},
\eeq
where $e\in G_0$ is the identity element.

Since the image of $G$ has a non-trivial GNVW index, to compute the 3-cocycle $\omega$ we stack the system of interest with its copy and let $G$ act on the second copy via a map which sends $(g,n)\in G$ to the automorphism $\tau^{-n}$. The action of $G$ on the composite system is via QCAs which have zero GNVW index and thus are circuits. Explicitly, the generator of the translation symmetry acts on the composite system by $\tau\otimes\tau^{-1}$, and one can write $\tau\otimes\tau^{-1}={\tilde S} S,$ where $S$ is the swap of the first and the second copy of the system (Fig. \ref{fig:shiftS}) and ${\tilde S}$ is the swap of the first copy and the second copy shifted to the right (Fig. \ref{fig:shiftStilde}). 

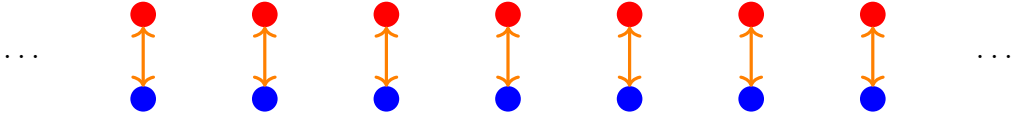
\begin{figure}
\begin{center}
\begin{tikzpicture}
    \begin{scope}[scale=1.6,yshift=0cm]
    \foreach \x in {-3,...,3}{
    \fill[blue] (\x,0) circle [radius=3pt];
    \fill[red] (\x,0.7) circle [radius=3pt];
   \draw[<->,orange,very thick] (\x,0.1)--(\x,0.6);
    }
    \node at (-4,0.35) {$\ldots$};
    \node at (4,0.35) {$\ldots$};
    \end{scope}
\end{tikzpicture}
\end{center}
\caption{The block-partitioned unitary $S$ which swaps the original system (blue) and its copy (red).}\label{fig:shiftS}
\end{figure}

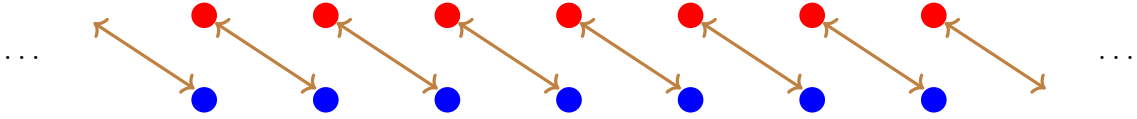
\begin{figure}
\begin{center}
\begin{tikzpicture}
    \begin{scope}[scale=1.6,yshift=0cm]
    \foreach \x in {-3,...,3}{
    \fill[blue] (\x,0) circle [radius=3pt];
    \fill[red] (\x,0.7) circle [radius=3pt];
   \draw[<->,brown,very thick] (\x+0.92,0.09)--(\x+0.09,0.64);
    }
    \draw[<->,brown,very thick] (-4+0.92,0.09)--(-4+0.09,0.64);
     \node at (-4.5,0.35) {$\ldots$};
     \node at (4.5,0.35) {$\ldots$};
    \end{scope}
\end{tikzpicture}
\end{center}
\caption{The circuit $\tilde S$.}\label{fig:shiftStilde}
\end{figure}

Thus $(g,n)\in G_0\times\ZZ$ acts on the composite system by a circuit 
\beq\label{eq:alphaLSM}
\alpha(g,n)=({\tilde S} S)^n\circ \left(\prod_{j\in\ZZ} \Ad_{\bpi_j(g)}\otimes\Id\right).
\eeq

A natural choice of an automorphism $\beta(g,n)\in \alLPAp$ such that $\alpha(g,n)\beta(g,n)^{-1}\in\alLPAm$ is a circuit 
\beq\label{eq:betaLSM}
\beta(g,n)=({\tilde S}_+ S_+)^n\circ \left(\prod_{j=1}^\infty \Ad_{\bpi_j(g)}\otimes \Id\right).
\eeq
where $S_+$ and ${\tilde S}_+$ are ``partial swaps'' whose action is schematically shown in Fig. \ref{fig:tSplusSplus}. For this choice of $\beta(g,n)$, one gets (up to scalar multiples):  $\CV(g,0;g',0)=\CV(g,1;g',0)=1$, $\CV(g,0;g',1)=\bpi_1(g)\otimes 1$. Substituting into (\ref{slantproduct}), we find
\beq
(\omega/[1])(g,g')=\rho(g,g').
\eeq
Thus the anomaly index for the $G_0\times\ZZ$-action regarded as an element of $H^2(G_0,U(1))$ is the class $[\rho]$ determined by the projective action of $G_0$ on the on-site Hilbert space.

\begin{figure}
\begin{center}
\begin{tikzpicture}
    \begin{scope}[scale=1.6,yshift=0cm]
    \foreach \x in {-3,...,3}{
    \fill[blue] (\x,0) circle [radius=3pt];
    \fill[red] (\x,0.7) circle [radius=3pt];
    }
    \foreach \x in {1,...,3}{
    \draw[<->,orange,very thick] (\x,0.1)--(\x,0.6);
    }
     \foreach \x in {0,...,3}{
      \draw[<->,brown,very thick] (\x+0.92,0.09)--(\x+0.09,0.64);
     }
      \node at (-4,0.35) {$\ldots$};
    \node at (4.5,0.35) {$\ldots$};  
    \end{scope}    
\end{tikzpicture}
\end{center}
\caption{The circuit ${\tilde S}_+ S_+.$ The orange arrows are implemented first.}\label{fig:tSplusSplus}
\end{figure}
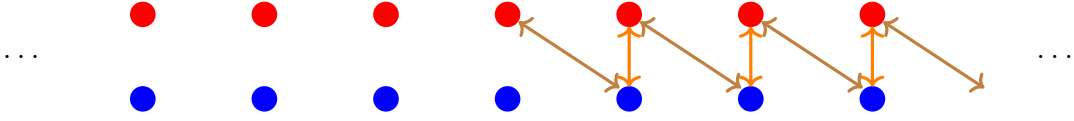

\end{example}

\section{Absence of symmetric gapped states for systems with an anomalous symmetry.} \label{sec:LSM}

A state of $\SAql$ is a positive linear functional $\SAql\ra \CCC$ of norm 1. A state is pure if it cannot be written as a convex linear combination of two different states in a nontrivial way. Equivalently, a state $\psi$ is pure if any positive linear functional $\rho:\SAql\ra\CCC$ which is majorized by $\psi$ (i.e. $\psi-\rho$ is positive) must be proportional to $\psi$. We denote the evaluation of a state $\psi$ on an observable $\CA \in \SA$ by $\lal \CA \ral_\psi \in \CCC$. Given a state $\psi$, we can produce a new state $\psi \circ \alpha$ by pre-composing it with an automorphism $\alpha$. That is,  $\lal \CA \ral_{\psi \circ \alpha} := \lal \alpha(\CA) \ral_\psi$.

For any state $\psi$, we let $(\CH_{\psi},\pi_{\psi},|\psi\ral)$ be its Gelfand–Naimark–Segal (GNS) triple: a Hilbert space $\CH_{\psi}$, a $*$-representation $\pi_{\psi}:\SAql \to \BoundedOper(\CH_{\psi})$, and a cyclic vector $|\psi\ral \in \CH_{\psi}$ such that $\lal \CA \ral_{\psi} = \lal \psi| \pi_{\psi}(\CA) | \psi\ral$. Such a triple is unique up to an isomorphism. The GNS representation $\pi_\psi$ is irreducible iff $\psi$ is pure \cite{bratteli2012operator}.

Two states $\psi_1$ and $\psi_2$ of a $C^*$-algebra are called unitarily equivalent if the corresponding GNS representations are unitarily equivalent. In this case we will write $\psi_1\sim\psi_2$. For example, for any $\CU\in \SUql$ and any state $\psi$ of $\SAql$, the states $\psi$ and $\psi\circ\Ad_\CU$ are unitarily equivalent. When $\psi_1\sim\psi_2$, one can identify  $\CH_{\psi_1}$ and $\CH_{\psi_2}$ and represent $\psi_2$ by a vector $|\psi_2\ral\in \CH_{\psi_1}$, in the sense that $\langle\CA\rangle_{\psi_2}=\langle\psi_2| \pi_{\psi_1}(\CA)|\psi_2\rangle$ for all $\CA\in\SAql$.

A finite-range Hamiltonian $\derH$ of a quantum spin system is an unbounded densely-defined derivation of the $C^*$-algebra $\SAql$ such that for any $\CA \in \SAl$ one has 
\beq
\derH(\CA) = \sum_{j \in \ZZ} [\chh_j,\CA]
\eeq
for some $R > 0$ and a collection of observables $\{\chh_j \in \SA_{\Ball_j(R)}\}_{j \in \ZZ}$ of uniformly bounded norm. We denote the vector space of such Hamiltonians by $\mfkDl$. Recall that a state $\psi$ is a ground state of $\derH$ if for any $\CA \in \SAl$ we have $\lal \CA^* \derH(\CA) \ral_{\psi} \geq 0$. This condition implies that $\psi$ is invariant under the one-parameter group of automorphisms $\alpha_\derH(t)$ generated by $\derH$ (which exists thanks to the Lieb-Robinson bound \cite{bratteli2012operator2}), that $|\psi\ral$ is invariant under the corresponding one-parameter group of unitaries $e^{it\hat H}$ acting on $\CH_\psi$, and that the spectrum of the Hamiltonian operator $\hat H$ on $\CH_\psi$ is non-negative. A ground state is called gapped if there is $\Delta > 0$ such that for any $\CA \in \SAl$ we have $\lal \CA^* \derH(\CA) \ral_{\psi} \geq \Delta \l \lal \CA^* \CA \ral_{\psi} - |\lal \CA \ral_{\psi}|^2 \r$.
This condition is equivalent to requiring the spectrum of $\hat H$ on the orthogonal complement to $|\psi\ral$ to be contained in $[\Delta,+\infty)$. In particular,  $|\psi\ral$ is the unique eigenvector of $\hat H$ with eigenvalue $0$. It is also not hard to show that every gapped ground state is pure.\footnote{Suppose $\psi$ is a gapped ground state of $\derH$ and $\rho:\SAql\ra\CCC$ is a nonzero positive linear functional majorized by $\psi$. Then there exists $T\in B(\CH_\psi)$ such that  $\rho(\CB^*\CA)=\langle\psi|\pi_\psi(\CB)^* T\pi_\psi(\CA)|\psi\rangle$ for any $\CA,\CB\in\SAql$  (see proof of Theorem 2.3.19 of \cite{bratteli2012operator}). This implies $T\in \pi_\psi(\SAql)'$, and since $e^{it\hat H}\in \pi_\psi(\SAql)''$ (Prop. 5.3.19 of \cite{bratteli2012operator2}), $T$ commutes with $e^{it\hat H}$. Let $|\psi'\rangle$ be the projection of the nonzero vector $T|\psi\rangle$ to the orthogonal complement of $|\psi\rangle$. Since $|\psi\rangle$ is the unique vector in $\CH_\psi$ which is invariant under $e^{it\hat H}$, $|\psi'\rangle$ must vanish. Therefore $T|\psi\rangle$ is proportional to $|\psi\rangle$, and thus $\rho$ is proportional to $\psi$.}

Although $\SAql\simeq \SAql_{<0}\otimes\SAql_{\geq 0}$, a pure state $\psi$ of $\SAql$ may not be unitarily equivalent to a state of the form $\psi_{<0}\otimes\psi_{\geq 0}$, where $\psi_{<0}$ (resp. $\psi_{\geq 0}$) is some pure state of $\SAql_{<0}$ (resp. $\SAql_{\geq 0}$). If such $\psi_{<0}$ and $\psi_{\geq 0}$ do exist, we say that $\psi$ satisfies the split property\footnote{In general, a not-necessarily-pure state $\psi$ is said to satisfy the split property if it is quasi-equivalent to a state of the form $\psi_{<0}\otimes\psi_{\geq 0}$ for some half-chain states $\psi_{<0}$ and $\psi_{\geq 0}$ \cite{Matsui}. Two states are called quasi-equivalent if the corresponding GNS representations are quasi-equivalent. Two representations $\pi_1$ and $\pi_2$ are called quasi-equivalent if every sub-representation of $\pi_1$ contains a sub-representation unitarily equivalent to a sub-representation of $\pi_2$ and vice versa.}. We will use the following result obtained by T. Matsui using the area law for gapped ground states proved in \cite{hastings2007area}:

\begin{theorem} \label{lma:splitproperty}
(Matsui, \cite{Matsui}.) Let $\psi$ be a gapped ground state of some $\derH \in \mfkDl$. Then there are pure states $\psi_{<0}$, $\psi_{\geq 0}$ of the algebras $\SAql_{<0}$, $\SAql_{\geq 0}$, respectively, such that the state $\psi_{<0} \otimes \psi_{\geq 0}$ of $\SAql$ is unitarily equivalent to $\psi$.
\end{theorem}

From this we deduce

\begin{lemma}\label{lma:unitaryeq}
Let $\psi$ be a gapped ground state of a Hamiltonian $\derH \in \mfkDl$, and let $\alpha \in \LPA$ be an automorphism with a trivial GNVW index such that $\psi \circ \alpha$ is unitarily equivalent to $\psi$. Then for any choice of the decomposition $\alpha = \beta_- \beta_+$ such that $\beta_\pm \in \LPApm$, the states $\psi \circ \beta_{\pm}$ are unitarily equivalent to $\psi$.
\end{lemma}

\begin{proof}

By Lemma \ref{lma:alphadecomposition}, we can represent $\alpha$ as $\alpha_{< 0} \alpha_0  \alpha_{\geq 0}$ for some $\alpha_{<0} \in \LPAsm$, $\alpha_{0} \in \LPAz$, $\alpha_{\geq 0} \in \LPAsp$. By Theorem \ref{lma:splitproperty}, there are pure states $\omega_{<0}$, $\omega_{\geq 0}$ on the left and right half-chains, respectively, such that $\psi$ is unitarily equivalent to the tensor product $\omega = \omega_{<0} \otimes \omega_{\geq 0}$. It follows that $(\omega_{< 0} \circ \alpha_{<0}) \otimes (\omega_{\geq 0} \circ \alpha_{\geq 0})$ is unitarily equivalent to $\omega_{<0} \otimes \omega_{\geq 0}$, and therefore, the state $\omega_{\geq 0} \circ \alpha_{\geq 0}$ on $\SAql_+$ (resp. the state $\omega_{<0} \circ \alpha_{<0}$ on $\SAql_-$) is unitarily equivalent to $\omega_{\geq 0}$ (resp. $\omega_{<0}$). Hence $\omega \circ (\Id\otimes\alpha_{\geq 0})$ is unitarily equivalent to $\omega$ and therefore $\psi \circ (\Id\otimes\alpha_{\geq 0})$ is unitarily equivalent to $\psi$. Since any $\beta_+$ arising from a decomposition $\alpha=\beta_- \beta_+$ has the form  $\Ad_\CU \alpha_{\geq 0}$ for some $\CU\in\SUql$, this implies the statement of the lemma.
\end{proof}

We say that $\derH \in \mfkDl$ is invariant under $\alpha \in \LPA$ if for any $\CA \in \SAql$, we have $\alpha(\derH(\CA)) = \derH(\alpha(\CA))$. If $\alpha: G \to \LPA$, then $\derH$ is $G$-invariant if it is invariant under $\alpha(g)$ for any $g \in G$.

\begin{theorem} \label{thm:gLSM}
Suppose a group $G$ acts on a 1d quantum lattice system $\SAql$ via a homomorphism $\alpha: G \to \LPA$. Suppose also that there exists a $G$-invariant finite-range Hamiltonian $\derH$ whose ground state  $\psi:\SAql\ra \CCC$ is $G$-invariant and gapped. Then the anomaly index $[\omega] \in H^3(G,U(1))$ vanishes. 
\end{theorem}
\begin{proof}
Suppose first that all automorphisms in the image of $\alpha$ have a trivial GNVW index. For every $g\in G$, we choose some $\beta(g)\in \LPAp$ such that  $\alpha(g)\beta(g)^{-1}\in \LPAm$. By Lemma \ref{lma:unitaryeq}, there exist unitary operators $U(g) \in \UnitaryOper(\CH_{\psi})$ implementing the automorphisms $\beta(g)$:
\begin{equation}
    \pi_{\psi}(\beta(g)(\CA))=U(g)\pi_{\psi}(\CA) U(g)^{-1},\quad\forall g\in G,\forall \CA\in\SA.
\end{equation}
Then for any choice of $\CV$,
\begin{equation}    \pi_{\psi}\left(\CV(g,h)\CA\CV(g,h)^{-1}\right)=U(g)U(h)U(gh)^{-1}\pi_{\psi}(\CA)U(gh)U(h)^{-1}U(g)^{-1},\quad \forall g,h\in G,\forall\CA\in\SA.
\end{equation}
Thus $\pi_\psi(\CV(g,h))U(gh)U(h)^{-1} U(g)^{-1}$ is in the commutant of $\pi_\psi(\SAal)$ and thus in the commutant of $\pi_\psi(\SAql)$. Since $\psi$ is pure, this implies there exists a function $\kappa:G\times G\ra U(1)$ such that 
\begin{equation}
    \pi_{\psi}(\CV(g,h))=\kappa(g,h)U(g)U(h)U(gh)^{-1}.
\end{equation}
Then a direct computation shows that $\omega=\pi_{\psi}(\omega\cdot 1_\SA)=\pi_{\psi}(\delta\kappa\cdot 1_\SA)=\delta\kappa$. Thus, the anomaly index is trivial.

If some automorphisms in the image of $\alpha$ have a non-trivial GNVW index, we stack the system with its copy on which $G$ acts by generalized translations so that the action on the composite is by locality-preserving automorphisms with a vanishing GNVW index. Such an action of $G$ always exists, as explained in the definition of the anomaly index. It is easy to see that for any isomorphism $\HilbV \to  \left(\bigotimes_{i \in J} \HilbV_{p_i}\right)\otimes \HilbV'$ there exists a gapped finite-range (in fact, zero-range) Hamiltonian $\derH_0$ on the second copy which is invariant under arbitrary generalized translations. By taking the Hamiltonian for the composite to be $\derH\otimes 1+1\otimes\derH_0$, we reduce to the case when all automorphisms in the image of $\alpha$ have a trivial GNVW index.
\end{proof}

Thus, if a finite-range Hamiltonian $\derH$ is invariant under an anomalous symmetry $G$, then every ground state of $\derH$ either spontaneously breaks $G$ or is gapless.

\begin{remark}
Theorem \ref{thm:gLSM} holds for any $G$-invariant pure state $\psi$ satisfying the split property, not just for gapped ground states of finite-range Hamiltonians. For example, it was shown in \cite{brandao2015exponential} that any translationally-invariant pure state which has exponential decay of correlations has bounded entanglement entropy and therefore has the split property (see Theorem 3.1 in \cite{Matsui} for a precise statement).
\end{remark}

\begin{remark}
Theorem \ref{thm:gLSM} has implications for  finite-volume systems. As explained in Section 2.2 of \cite{tasaki2022lieb}, a sequence of finite systems of increasing size with finite-range Hamiltonians giving $\derH \in \mfkDl$ in the thermodynamic limit and with $G$-invariant ground states that are separated by a uniform (i.e. independent of the size) energy gap from the rest of the spectrum defines (at least one) $G$-invariant gapped ground state for $\derH$. Since according to  Theorem \ref{thm:gLSM} $G$-invariant gapped ground states for $\derH$ do not exist if the symmetry is anomalous, it follows that no such sequence of finite systems can exist either.
\end{remark}

\begin{example} Consider the anomalous $\ZZ/2$ action from Example  \ref{example:z2}. The generator of the symmetry maps the trivially gapped Hamiltonian (written in a physics convention as a formal sum of self-adjoint observables) 
\begin{equation}
H_0=-\sum_j X_j
\end{equation}
to 
\begin{equation}
H_1=-\sum_j X_j Z_{j-1}Z_{j+1}.
\end{equation}
Thus the derivation $H_0+H_1$  has an anomalous $\ZZ/2$ symmetry, and the ground state of the corresponding derivation $\derH_0+\derH_1=i\ad_{H_0}+i\ad_{H_1}$ is either gapless or breaks $\ZZ/2$ spontaneously. The conclusion is unchanged if we deform $H_0+H_1$ by adding a nearest-neighbor Ising interaction
\begin{equation}
H_J=-J\sum_j Z_j Z_{j+1}.   
\end{equation}
Using the Jordan-Wigner transformation one can map this model to a model of non-interacting fermions and verify that as one increases $J$ from $0$ to sufficiently large positive values, there is a phase transition from a gapless phase described by $c=1$ CFT to a gapped phase with a spontaneously broken $\ZZ/2$ \cite{clusterstate1,clusterstate2,clusterstate3,clusterstate4}.

In addition to the $\ZZ/2$ symmetry, the family of models with the Hamiltonian $H_0+H_1+H_J$ has an anti-unitary time-reversal symmetry $(\ZZ/2)^T$. The phase diagram described above has been interpreted as a consequence of a mixed anomaly between the unitary $\ZZ/2$ and the anti-unitary $(\ZZ/2)^T$ \cite{pivot}. We see here that the unitary $\ZZ/2$ is anomalous by itself, so the presence of $(\ZZ/2)^T$ is inessential. The same behavior should be observed for Hamiltonians which break $(\ZZ/2)^T$, such as
\beq
H=H_0+H_1+H_J+ a \sum_j Y_j (1-Z_{j-1}Z_{j+1}).
\eeq
Here $Y_j=iX_j Z_j$. This model cannot be mapped to a model of non-interacting fermions.

\end{example}

\begin{example}
Consider a projective on-site action of a finite group $G_0$ together with translations. As shown in Example \ref{example:LSMfinite}, in this case there is a ``mixed anomaly'' between the two, so a finite-range Hamiltonian invariant under $G=G_0\times\ZZ$ cannot have a unique gapped ground state. This is a version of the LSM theorem proved in \cite{ogatatachikawatasaki} by a different method.
\end{example}

\appendix

\section{Anomaly index for a Lie group symmetry} \label{app:LieGroup}

The goal of this appendix is to generalize the discussion in the main text to the case of a Lie group symmetry. There are several modifications one needs to make. First of all, to define a Lie group symmetry of a spin system, one needs to use a suitable class of automorphisms of the algebra of observables which forms a group and such that the notion of a smooth homomorphism is well-defined. One such class was introduced in \cite{LocalNoether}. Secondly, as we will see, one can define  anomaly indices for Lie groups which take values in differentiable group cohomology introduced in \cite{brylinski2000differentiable} rather than the usual group cohomology.\footnote{It was suggested in Ref. \cite{ogatatachikawatasaki} that Borel group cohomology introduced in Ref. \cite{Mackey,Moore} is a suitable receptacle for anomaly indices of topological group symmetries of lattice systems. In the context of indices of SPT phases, the use of Borel group cohomology was advocated in Refs. \cite{chen2013symmetry,DQ}. The physical meaning of Borel cocycles is unclear to us.} This is quite useful, since differentiable group cohomology of Lie groups is much easier to compute.

\subsection{Differentiable group cohomology of Lie groups}

Recall that a simplicial set $S_{\bullet}$ is a collection of sets $\{S_n\}_{n \in \NN_0}$ together with maps $d_{k}:S_{n} \to S_{n-1}$, $n>0$, $0\leq k\leq n$ and $s_{k}:S_n \to S_{n+1}$, $n \geq 0$, $0\leq k \leq n$ satisfying
\begin{align}
& d_j \circ d_k = d_{k-1} \circ d_j\ \ \text{for}\ j<k,\nonumber\\
& s_j \circ s_k = s_{k} \circ s_{j-1}\ \ \text{for}\ j>k,\\ \nonumber
& d_j \circ s_k = 
\begin{cases}
    s_{k-1} \circ d_{j}\ \ \text{for}\ j<k\\
    \Id\ \ \text{for}\ j \in \{k,k+1\}\\
    s_{k} \circ d_{j-1}\ \ \text{for}\ j>k+1.
\end{cases}
\end{align}
The set $S_n$ is called the set of $n$-simplices. The maps $d_k$ and $s_k$ are called the face and the degeneracy maps, respectively. For a leisurely introduction to simplicial sets the reader may consult \cite{Friedman}.

Similarly, a simplicial manifold $M_{\bullet}$ is a collection of manifolds $\{M_n\}_{n \in \NN_0}$ with the same maps and relations which are required to be smooth. If $\{U^{(a)}\}_{a \in I}$ is a good open cover of a manifold $X$, then $\{ U_n = \bigsqcup_{a_0,...,a_n} U^{(a_0)} \cup ... \cup U^{(a_n)}\}_{n \in \NN_0}$ gives an example of a simplicial manifold with the maps $d_k$, $s_k$ being natural submersions and embeddings.

Let $G$ be a group (with discrete topology). There is an associated simplicial set $B_{\bullet} G$ whose set of $n$-simplices $B_n G$ is given by $G^n$ and whose face $d_k:G^n \to G^{n-1}$, $k=0,...,n$ and degeneracy $s_k:G^{n} \to G^{n+1}$, $k=0,...,n$ maps are defined by
\begin{align}
& d_k (g_1,...,g_n) = 
\begin{cases}
    (g_2,...,g_n)\ \ \text{for}\ k=0\\
    (g_1,...,g_k g_{k+1},...,g_n)\ \ \text{for}\ 0<k<n\\
    (g_1,...,g_{n-1})\ \ \text{for}\ k=n,
\end{cases}
\\
& s_k (g_1,...,g_n) = 
\begin{cases}
    (\Id,g_1,...,g_n)\ \ \text{for}\ k=0\\
    (g_1,...,g_k, \Id, g_{k+1},...,g_n)\ \ \text{for}\ 0<k<n\\
    (g_1,...,g_{n},\Id)\ \ \text{for}\ k=n.
\end{cases}
\end{align}
For an abelian group $A$ (with discrete topology), we can form a cochain complex $C^{\bullet}(B_{\bullet}G,A)$ with $C^{k}(B_{\bullet}G,A)$ being arbitrary maps $f:G^k \to A$ and the differential being the alternating sum $d f := d^*_0 f - d^*_1 f + ... + (-1)^n d^*_n f$. The cohomology $H^{\bullet}(B_{\bullet}G,A)$ of this complex can be taken as the definition of the group cohomology $H^\bullet(G,A)$.

For a Lie group $G$ and an abelian Lie group $A$, there is a version of group cohomology with coefficients in $A$ defined in \cite{brylinski2000differentiable} which is more suitable for our purposes and coincides with the usual group cohomology when $G$ is discrete. We now let $B_{\bullet} G$ be a simplicial manifold whose manifold of $n$-simplices $B_n G$ is given by $G^n$ and whose face and degeneracy maps are defined in the same way as above. To define differentiable group cohomology theory, we need to choose a simplicial cover of $B_{\bullet} G$. It consists of a simplicial set $I_{\bullet}$ and open covers $\{U^{(a)}_p\}_{a \in I_p}$ of $G^p$ for each $p>0$ such that $d_k U^{(a)}_p \subseteq U^{(d_k a)}_{p-1}$ and $s_k U^{(a)}_p \subseteq U^{(s_k a)}_{p+1}$. Such a cover exists for any simplicial manifold \cite{BRMCL1}. We denote by $U_{p,\bullet}$ the simplicial manifold which corresponds to the cover of $G^p$ with
\beq
U_{p,q} = \bigsqcup_{a_0,...,a_q \in I_p} (U^{(a_0)}_{p} \cap ... \cap U^{(a_q)}_{p})
\eeq
and face and degeneracy maps denoted by $\delta$ and $\sigma$, respectively. We denote the collection of manifolds $\{U_{p,q}\}_{p,q \in \NN_0}$ together with the maps $d_k$, $s_k$, $\delta_k$, $\sigma_k$ by\footnote{Such a collection forms a bi-simplicial manifold.} $U_{\bullet,\bullet}$. We then define a double complex with values in $A$ as  
\beq
C^{p,q}(U_{\bullet,\bullet}, \underline{A}) := \bigoplus_{a_0,...,a_q \in I_{p}} C^{\infty}( U^{(a_0)}_{p} \cap ... \cap U^{(a_q)}_{p}, A).
\eeq
The differentials are defined by the alternating sums $d f := d^*_0 f - d^*_1 f + ... + (-1)^n d^*_n f$ and $\delta f := \delta^*_0 f - \delta^*_1 f + ... + (-1)^n \delta^*_n f$. 

Let $H^n(U_{\bullet,\bullet},A)$ be the total cohomology of this double complex in degree $n$. While in general $H^n(U_{\bullet,\bullet},A)$ depends on the choice of the simplicial cover,  varying $U_{\bullet,\bullet}$ one gets a directed system of groups labeled by the directed set of simplicial covers, and by taking the limit over all simplicial covers one gets an abelian group which we denote $H^n_{diff}(G,A)$. Alternatively, one can choose $U_{\bullet,\bullet}$ to be a good simplicial cover (i.e. a simplicial cover such that all $U_{p,q}$ are contractible). Then $H^n(U_{\bullet,\bullet},A)$ is independent of the choice of $U_{\bullet,\bullet}$ and is isomorphic to $H^n_{diff}(G,A)$ \cite{BRMCL1}. 

Clearly, when $G$ is discrete and $A$ is arbitrary, we have $H^\bullet_{diff}(G,A)=H^\bullet(G,A)$. On the other hand, if $G$ is arbitrary and $A$ is discrete, we have $H^\bullet_{diff}(G,A)=H_{sing}^\bullet(BG,A)$, where $BG$ is the classifying space of $G$-bundles. 

In this paper we are interested in the case $A=U(1)$ and $G$ arbitrary. To compute $H^n_{diff}(G,A)$ in this case, one can try using the long exact sequence in differentiable cohomology induced by the coefficient short exact sequence\footnote{Since this is a complex of abelian Lie groups rather than a complex  of abstract abelian groups, the notion of an exact sequence is defined differently. Namely, one says that a complex of abelian Lie groups is an exact sequence if it is exact as a sequence of abstract abelian groups as well as a smooth locally-trivial fiber bundle  \cite{brylinski2000differentiable}.} $0\ra \ZZ\ra\RR\ra U(1)\ra 0$. The corresponding long exact sequence implies an isomorphism $H^n_{diff}(G,U(1))\simeq H^{n+1}_{diff}(G,\ZZ)$ for all $n>0$ provided one knows that $H^n_{diff}(G,\RR)=0$ for all $n>0$. This is so if $G$ is compact or at least if $G$ has a compact Lie subgroup with the same Lie algebra as $G$ \cite{Mostow}. Thus, if $G$ is compact we get  $H^n_{diff}(G,U(1))\simeq H^{n+1}_{diff}(G,\ZZ)\simeq H_{sing}^{n+1}(BG,\ZZ).$ In particular, for $G=U(1)^k$, we can take $BG$ to be $\left(\CCC{\mathbb P}^\infty\right)^k$, and thus $H^n_{diff}(U(1)^k,U(1))\simeq H^{n+1}_{sing}(BU(1)^k,\ZZ)$ vanishes for all positive even $n$. 

Another case of interest is $G=G_0\times\ZZ$, where $G_0$ is a compact Lie group. Again one has $H^n_{diff}(G_0\times\ZZ,U(1))\simeq H^{n+1}_{sing}(BG_0\times B\ZZ,\ZZ)$ for all $n>0$, and therefore from the universal coefficient formula and the K\"unneth formula we get 
\beq
H^n_{diff}(G_0\times\ZZ,U(1))\simeq H_{sing}^{n+1}(BG_0,\ZZ)\oplus H_{sing}^n(BG_0,\ZZ)\simeq H^n_{diff}(G_0,U(1))\oplus H^{n-1}_{diff}(G_0,U(1)).
\eeq

The group $H^2_{diff}(G,U(1))$ has a simple meaning: its elements label central extensions of $G$ by $U(1)$ such that the surjection to $G$ is a smooth locally-trivial fiber bundle \cite{brylinski2000differentiable}.

\subsection{Almost local observables}

The algebra of local observables $\SAl$ can be completed with respect to a countable family of seminorms introduced in \cite{LocalNoether}:
\beq
\|\CA\|_{\alpha,j} :=\|\CA\|+ \sup_{r \in \RR_{\geq 0}} (1+r)^{\alpha} \inf_{\CB \in \SA_{\Ball_j(r)}} \|\CA - \CB\|,
\eeq
where $\alpha \in \NN$, $j \in \Lambda$ and $\Ball_j(r)$ is the interval $[j-r,j+r] \subset \RR$. The completion is a \Frechet\ $*$-algebra $\SAal$ of {\it almost local observables} . Loosely speaking, elements of $\SAal$ can be approximated in the norm by elements of $\SAl$ localized on balls of radius $r$ with an $\Or$ error. More precisely, for any $\CA\in\SAal$ and any $j\in\Lambda$ there exists a monotonically decreasing function $f_j:\RRp \to \RRp$ which is $\Or$ and satisfies $\|\CA-\CB\|\leq f_j(r)$ for some $\CB\in\SA_{\Ball_j(r)}$. We then say that $\CA$ is $f_j$-localized on $j$. A canonical choice for $f_j(r)$ is
\beq
f_j(\CA,r)=\inf_{\CB \in \SA_{\Ball_j(r)}} \|\CA - \CB\|.
\eeq
In contrast to $\SAql$, $\SAal$ depends on the metric on the lattice (induced by the embedding of $\ZZ$ into the Euclidean line). We denote the group of unitary elements of $\SAal$ by $\SUal$. 

For a subset $X \subset \RR$, the algebra $\SAal_X$ is defined as the completion of $\SAl_X$ with respect to the seminorms
\beq
\|\CA\|_{X,\alpha,j} :=\|\CA\|+ \sup_{r \in \RR_{\geq 0}} (1+r)^{\alpha} \inf_{\CB \in \SA_{\Ball_j(r) \cap X}} \|\CA - \CB\|,
\eeq
When $X = \RR_{\geq 0}$ (resp. $\RR_{<0}$), we denote $\SAal_{X}$ by $\SAal_{\geq 0}$ (resp. $\SAal_{<0}$). The group of unitary elements in $\SAal_X$ will be denoted $\SUal_X$.

\subsection{Almost local derivations}

An infinitesimal action of a Lie group symmetry corresponds to (unbounded skew-symmetric) derivations of the $C^*$-algebra $\SAql$, i.e. densely defined linear maps $\derH: \SAql \to \SAql$ which commute with $*$ and satisfy on their domain of definition the Leibniz rule $\derH(\CA \CB) = \derH(\CA) \CB + \CA \derH(\CB)$. Such derivations should form a Lie algebra and generate (non-infinitesimal) action of symmetry via integration. A class of derivations with a common dense domain which has these desirable properties was introduced in \cite{LocalNoether}. We review the definition below.

Let $\mathbb{B}_1$ be the set $\{[n,m]\subset \RR|\,n,m \in \ZZ, \,n\leq m 
\} \cup \{\emptyset\}$, and let $\text{diam}(Y)$ be the diameter of $Y \subset \RR$. A derivation $\derH$ is called finite-range if there exist $R>0$ and a collection of local anti-self-adjoint observables $\{\chh^{Y} \in \SAl_{Y}\}_{Y \in \mathbb{B}_1,\diam(Y)\leq R}$ such that $\|\chh^{Y}\| \leq C$ for some $C >0$ and for any $\CA \in \SAl$ we have
\beq
\derH(\CA) = \sum_{Y} [\chh^{Y}, \CA].
\eeq
The choice of $\chh^Y$ for a given derivation $\derH$ is not unique. However, as explained in \cite{LocalNoether}, there is a canonical choice of $\chh^{Y}$ defined by the condition $\lal \CA^* \chh^{Y} \ral_{\infty} = 0$ for any $\CA \in \SA_{Z}$, $Z \in \mathbb{B}_1$ such that $Z \subsetneq Y$, where $\lal\,\cdot\,\ral_{\infty}:\SAql \to \CCC$ is the unique tracial state of $\SAql$ (the infinite-temperature state). We denote such $\chh^Y$ by $\derH^Y$. We denote the space of finite-range derivations by $\mfkDl$. $\mfkDl$ is a real Lie algebra with the bracket $[\,\cdot\,,\,\cdot\,]:\mfkDl \times \mfkDl \to \mfkDl$ defined by $[\derH,\derG](\CA):= \derH(\derG(\CA)) - \derG(\derH(\CA))$ for $\CA \in \SAl$. 

We can complete the space $\mfkDl$ with respect to the norms
\beq \label{eq:Dnorms}
\|\derH\|_{\alpha} := \sup_{Y \in \mathbb{B}_1} (1+\text{diam}(Y))^{\alpha} \|\derH^{Y}\|.
\eeq
The completion $\mfkDal$ is a \Frechet\ space. We call its elements {\it almost local derivations}. It follows from the above definition that for any $\derH\in\mfkDal$ there exists a positive non-increasing function $f(r)=\Or$ such that for any $Y\in \Br_1$, $\|\derH^Y\|\leq f(\diam(Y))$. In this case we will say that $\derH$ is $f$-localized. 

In \cite{LocalNoether} various properties of almost local derivations are discussed. In particular, they are defined everywhere on $\SAal$, map $\SAal$ to itself, and the Lie bracket extends to $\mfkDal$ giving it the structure of a real \Frechet-Lie algebra.

For future use, we also introduce \Frechet-Lie sub-algebras $\mfkDalsp$, $\mfkDalsm$ defined as follows. The space $\mfkDalsp$ is the completion in the norms eq. (\ref{eq:Dnorms}) of a subspace of $\mfkDl$ that consists of derivations $\derH$ such that $\derH^Y = 0$ for $Y \not\subseteq \RRp$. Clearly, $\mfkDalsp$ is a closed sub-algebra of $\mfkDal$. The closed sub-algebra $\mfkDalsm$ is defined similarly, with $\RRp$ replaced with $\RRm$.

Let  $\mfkdal$ be a real closed subspace of $\SAal$ consisting of observables satisfying $\CA=-\CA^*$ and $\lal \CA \ral_{\infty}=0$. Such observables define derivations of $\SAal$ via the commutator.  Any derivation $\derH \in \mfkDal$ admits a canonical decomposition $\derH = \derH_{<0}+ \ad_{\chh_0} + \derH_{\geq 0}$ where $\derH^Y_{<0}$ (resp. $\derH^Y_{\geq 0}$) coincides with $\derH^Y$ for $Y \subset \RR_{<0}$ (resp. $Y \subset \RR_{\geq 0}$) and vanishes otherwise, while $\chh_0\in\mfkdal$ is the sum of all $\derH^Y$ for $Y$ intersecting $-1/2$. Thus, $\mfkDal$ is canonically identified with a subspace of $\mfkDal_{<0}\oplus\mfkDal_{\geq 0}\oplus\mfkdal$. It is easy to see that it is a closed subspace. Moreover, for a continuous (resp. smooth) map $\derH:M \to \mfkDal$ the pointwise decomposition produces continuous (resp. smooth) maps $\derH_{<0}: M \to \mfkDalsm$, $\chh_0: M \to \mfkdal$, $\derH_{\geq 0}:M \to \mfkDalsp$. 
\begin{remark}
It follows from Prop. C.2 of Ref. \cite{LocalNoether} that all derivations of the form $\ad_\chh, \chh\in\mfkdal$,  belong to $\mfkDal$. Thus the map $\ad:\mfkdal\ra\mfkDal, \chh\mapsto\ad_\chh$ identifies the \Frechet -Lie algebra $\mfkdal$ with a sub-algebra of the \Frechet -Lie algebra $\mfkDal$. However, although $\ad$ is continuous, $\ad(\mfkdal)$ is not a closed sub-algebra of $\mfkDal$. 
\end{remark}

The following lemma is often useful:

\begin{lemma}\label{lma:inner}
    Let $\chh\in C^0([0,1],\mfkdal)$. For any $\CU_0\in\SUal$, the differential equation 
    \beq\label{eq:LGAinnerdef}
    \frac{d\CU(s)}{ds}=\CU(s)\chh(s)
    \eeq
    has a unique solution in $C^1([0,1],\SUal)$ with the initial condition $\CU(0)=\CU_0$. If $\chh(s)$ is $f$-localized on $0$ for any $s\in [0,1]$, then $\CU(s)$ is $g$-localized on $0$ for some $g(r)=\Or$ which depends only on $f$ and $\|\chh\|=\sup_{t\in [0,1]} \|\chh(t)\|$. If $\chh$ smoothly depends on parameters in $\RR^n$, then $\CU(s)$ is also a smooth function of these parameters for any $s\in [0,1]$.
\end{lemma}
\begin{proof}
    To establish the existence of a solution for any $\CU_0$ it is sufficient to do it for $\CU_0=1$. For any $n\in\NN$, let $\chh_n=\sum_{Y\in {\mathbb B}_1,Y\subset [-n,n]} \chh^Y$. For each $s\in [0,1]$, $\chh_n(s)$ is a local anti-self-adjoint observable supported on $[-n,n]$, and the sequence $\chh_n\in C^0([0,1],\mfkdal)$ converges to $\chh$. By the standard results from the theory of ODEs, for any $n$ there is a solution  $\CU_n\in C^1([0,1],\SUl_{[-n,n]})$ to the differential equation $d\CU_n/ds=\CU_n \chh_n$ with the initial condition $\CU_n(0)=1$. An equivalent integral equation is
    \beq\label{eq:integraleq}
    \CU_n(s)=1+\int_0^s \CU_n(t) \chh_n(t) dt.
    \eeq
    The differential equation for $\CU_n$ implies
    \beq
    \frac{d}{ds}\left(\CU_{n+1}\CU_n^{-1}\right)=\CU_{n+1}\left(\chh_{n+1}-\chh_n\right)\CU_n^{-1},
    \eeq
    integrating which we get
    \beq
    \CU_{n+1}(s)=\CU_n(s)\left(1+\int_0^s \CU_{n+1}(t)(\chh_{n+1}(t)-\chh_n(t)) \CU_n(t)^{-1} dt\right).
    \eeq
    Therefore 
    $$
    \|\CU_{n+1}-\CU_n\|=\sup_s \|\CU_{n+1}(s)-\CU_n(s)\|\leq \|\chh_{n+1}-\chh_n\|,
    $$
    and thus $\CU_n$ converges uniformly on $[0,1]$ to some $\CU\in C^0([0,1],\SUql)$. By passing to the limit $n\ra\infty$ in Eq. (\ref{eq:integraleq}), we conclude that
    \beq\label{eq:integraleqCU}
    \CU(s)=1+\int_0^s \CU(t)\chh(t) dt.
    \eeq
    Therefore $\CU$ is continuously differentiable and satisfies Eq. (\ref{eq:LGAinnerdef}) and the initial condition $\CU(0)=1$. By multiplying this solution by $\CU_0$, we get a solution for a general initial condition. Uniqueness follows from the differential equation satisfied by $\CU$ in the standard manner. 

    Now let us show that $\CU$ is $g$-localized on $0$ for some $g\in\Or$. Suppose that $\chh(s)$ is $f$-localized at $0$ for all $s$. For any $\CA\in\SAql$ and any $r\geq 0$ we let
    $f_j(\CA,r)=\sup_{\CB\in \SAl_{\Ball_j(r)}}\|\CA-\CB\|$. For a fixed $j$ and $r$, it is a continuous seminorm on $\SAql$. Further, it is easy to see that for all $r\geq 0$ one has $f_j(1,r)=0$ and $\forall \CA,\CA'$  $f_j(\CA\CA',r)\leq \frac{3}{2}(\|\CA\|f_j(\CA',r)+\|\CA'\|f_j(\CA,r))$. Then Eq. (\ref{eq:integraleq}) together with Jensen's inequality for Banach-valued functions implies
    \beq
    f_0(\CU(s),r)\leq \frac{3s}{2}f(r)+\frac{3}{2}\int_0^s \|\chh(t)\| f_0(\CU(t),r) dt.
    \eeq
    By the integral Gronwall inequality, we get 
    \beq
    f_0(\CU(s),r)\leq \frac{3s}{2}f(r)\exp\left(\frac{3}{2}\int_0^s \|\chh(t)\|dt\right)\leq \frac{3}{2}f(r)\exp\left(\frac{3}{2}\|\chh\|\right).
    \eeq

    Finally, smooth dependence of $\CU(s)$ on the parameters of $\chh$ follows from the integral equation (\ref{eq:integraleqCU}). 
    
\end{proof}

\subsection{Almost Local Automorphisms}

Next we define a class of automorphisms of $\SAql$ for which the notion of a smooth homomorphism makes sense. 
\begin{definition}
An almost local locality-preserving automorphism is as automorphism with $g(r)$-tails for some $g(r) = \Or$.
\end{definition}
An almost local LPA maps $\SAl$ to $\SAal$. Moreover, one can show that every such automorphism maps $\SAal$ to $\SAal$. More precisely, we have
\begin{lemma} \label{lma:alLPAmapsAal2Aal}
    If $\alpha\in\alLPA$ has $g(r)$-tails and $\CA$ is $f$-localized at $j$, then $\alpha(\CA)$ is $h$-localized at $j$ for $h=\Or$ which depends only on $f$, $g$ and $\|\CA\|$.
\end{lemma}
\begin{proof}
    For any $n\in\NN$ pick a $\CB_n\in\SA_{\Ball_j(n)}$ such that $\|\CA-\CB_n\|\leq f(n)$. Note that $\|\CB_n\|\leq f(n)+\|\CA\|$. Since for a fixed $j,r$ the map 
    \beq
    \CA\mapsto f_j(\CA,r)=\inf_{\CB\in\SA_{\Ball_j(r)}}\|\CA-\CB\|
    \eeq
    is a semi-norm satisfying $f_j(\CA,r)\leq \|\CA\|$, we get 
    \beq
    f_j(\alpha(\CA),r)\leq f_j(\alpha(\CB_n),r)+f(n).
    \eeq
    On the other hand, for any $r>n$ we have $f_j(\alpha(\CB_n),r)\leq g(r-n)\|\CB_n\|\leq g(r-n)(\|\CA\|+f(n))$.
    Therefore for $n=\lfloor r/2\rfloor$ we have 
    \beq\label{eq:hloc}
    f_j(\alpha(\CA),r)\leq g(r/2)\|\CA\|+f(\lfloor r/2\rfloor)(1+g(r/2)).
    \eeq
    We can take $h(r)$ to be the r.h.s. of (\ref{eq:hloc}).

\end{proof} 

Almost local LPAs form a subgroup of the group of all LPAs. We denote it $\alLPA$. The group of QCAs is contained in $\alLPA$. 

As explained in \cite{LocalNoether}, any $\derH \in \mfkDal$ can be represented by a collection of observables $\{\chh_j\}_{j \in \Lambda}$ each of which is $h$-localized on $j$ for some positive non-increasing function $h(r) = \Or$ such that for any $\CA \in \SAl$ we have $\derH(\CA) = \sum_{j \in \Lambda} [\chh_j,\CA]$. Using this representation and Lemma \ref{lma:alLPAmapsAal2Aal}, one can define an action of $\alLPA$ on $\mfkDal$. For an almost local LPA $\alpha$, the derivation $\alpha(\derH)$ is defined by $\alpha(\derH)(\CA) = \sum_j [\alpha(\chh_j),\CA]$ for any $\CA \in \SAl$. It is easy to see that this definition is independent of the representation $\{\chh_j\}_{j \in \Lambda}$.

In Section \ref{sec:LPAs} we mentioned that the GNVW index is, in a sense, the only obstruction for realizing an LPA by a Hamiltonian evolution. Let us make this precise for almost local LPAs. First, we need to describe a suitable class of Hamiltonian evolutions. 
\begin{definition}
Let $\derF\in C^0([0,1],\mfkDal)$. We say that a one-parameter family of automorphisms $\alpha(s):\SAal\ra\SAal$, $s\in [0,1]$, is generated by $\derF$ if $\alpha(0)=\Id$,  $\alpha(s)(\CA)$ is a differentiable function of $s$ for any $\CA\in\SAal$, and
\beq\label{eq:LGAdef}
\frac{d}{ds} \alpha(s)(\CA) = \alpha(s)(\derF(s)(\CA)),\quad \forall s\in [0,1].
\eeq
\end{definition}
It follows from the version of the Lieb-Robinson bound proved in \cite{HastingsKoma,nachtergaele2006propagation,nachtergaelesimsyoung} (see Prop. E.2 of \cite{LocalNoether}) that for any $\derF\in C^0([0,1],\mfkDal)$ there exists a unique $\alpha(s), s\in [0,1],$ which is generated by $\derF$. Given $\derF\in C^0([0,1],\mfkDal)$, we denote by $\alpha_\derF(s)$ the family of automorphisms generated by $\derF$.
\begin{definition}
We say that an automorphism $\beta:\SAal\ra\SAal$ is a Locally Generated Automorphism (LGA) if there exists $\derF\in C^0([0,1],\mfkDal)$ such that $\alpha=\alpha_\derF(1)$. In this case we also say that $\alpha$ is generated by $\derF$.
\end{definition}
\begin{remark}\label{rem:inner}
    Let $\chh\in C^0([0,1],\mfkdal)$. The LGA generated by the family of inner derivations $\derH={\rm ad}_\chh$ has the form $\Ad_{\CU(1)}$, where $\CU\in C^0([0,1],\SUal)$ is the unique solution of (\ref{eq:LGAinnerdef}) with the initial condition $\CU(0)=1$ (see Lemma \ref{lma:inner}). 
\end{remark}

It follows from Prop. E.1 in \cite{LocalNoether} that every LGA is an almost local LPA. More precisely, Prop. E.1 implies that if $\derF(s)$ is $f$-localized for all $s\in [0,1]$, then $\alpha_\derF(1)$ has $g(r)$-tails where $g(r)=\Or$ depends only on $f$. The continuity of the (generalized) GNVW index proved in \cite{ranard2022converse} implies that every LGA has a vanishing GNVW index. Conversely, we have the following result (which is a corollary of Theorems 5.6 and 5.8  of \cite{ranard2022converse}).
\begin{prop} \label{prop:GNVWforLGA}
Any almost local LPA $\beta$ with a vanishing GNVW index is locally generated.
\end{prop}
\begin{remark}
  LGAs form a group which is a subgroup of $\alLPA$. Indeed, if $\alpha_1(s)$ is generated by $\derF_1(s)$ and $\alpha_2(s)$ is generated by $\derF_2(s)$, then one can check that $\alpha_1(s)\alpha_2(s)$ is generated by $\derF_2(s)+\alpha_2(s)^{-1}(\derF_1(s))$ \cite{LocalNoether}. Obviously, every block-partitioned unitary is an LGA generated by a constant $\derF$. Therefore, every circuit (i.e. every zero-index QCA) is also an LGA. 
\end{remark}

\begin{proof}
Suppose the tails of $\beta$ are controlled by a function $g(r)=\Or$. By Theorem 5.6 of Ref. 
\cite{ranard2022converse} there exists a sequence of QCAs $\beta_m$, $m\in\NN$, of radius $2m$, such that for any $\CA\in\SA^l$ supported on a $X\subset\ZZ$ one has $\|\beta(\CA)-\beta_m(\CA)\|\leq \|\CA\|\frac{\diam (X)}{m} C_g  g(m)$ for some constant $C_g$ which depends on the function $g$.  Let $\gamma_1=\beta_1$ and $\gamma_m=\beta_m \beta_{m-1}^{-1}$ for $m>1$. Then $\gamma_m$ has radius at most $4m-2$ and satisfies $\|\gamma_m(\CA)-\CA\|\leq 2 \|\CA\|\frac{\diam (X)}{m} C_g  g(m)$ for any $\CA\in\SA^l$ supported on $X$. 

By Theorem 5.8 of \cite{ranard2022converse}, there exists $m_0>0$ such that for $m\geq m_0$ the GNVW-index of the QCA $\beta_m$ is the same as that of $\beta$, that is, zero. Hence the same is true about $\gamma_m$ for $m>m_0$. By Prop. 4.12 of \cite{ranard2022converse}, each $\gamma_m$, $m>m_0$, can be presented as a composition of two layers of block-partitioned unitaries of diameter no larger than $8m-4$, and each unitary is an exponential of an anti-self-adjoint element of $\SAl$ which has norm $O(g(m))$. That is, we can write $\gamma_m=\gamma_m^{(2)}\gamma_m^{(1)}$, where 
\begin{equation}
\gamma_m^{(a)}=\prod_{j=-\infty}^\infty {\rm Ad}_{\exp \left(\chh_{m,j}^{(a)}\right)},\quad a=1,2,
\end{equation}
is an automorphism of $\SA^l$, 
and $\chh_{m,j}^{(a)}$ are elements of $\SAl$ of diameter at most $8m-4$ whose supports for fixed $m,a$ and different $j$ do not overlap, and $\|\chh_{m,j}^{(a)}\|=O(g(m))$ uniformly in $m$.

Let $\gamma_m^{(a)}(s)$, $s\in\RR$, be the one-parameter family of  automorphisms generated by a constant finite-range derivation $\derH_m^{(a)}=\sum_j \chh_{m,j}^{(a)}$, so that $\gamma_m^{(a)}(1)=\gamma_m^{(a)}$. Also, let $\gamma_m(s)=\gamma_m^{(2)}(s)\gamma_m^{(1)}(s).$ Then it is easy to see that for any $s\in\RR$, $\gamma_m(s)$ maps $\SA^l$ to $\SA^l$ and increases the radius of support of a local observable by at most $16m-8$. Further, $\gamma_m(s)$ is generated by a continuous family of finite-range derivations 
\begin{equation}
\derG_m(s)=\derH_m^{(1)}(s)+\gamma^{(1)}_m(s)^{-1}(\derH_m^{(2)}(s)),
\end{equation}
of range at most $40m-20$. Note also $\|\derG_m\|_0=O(g(m))$.

For any $n\in\NN$ and any $s\in\RR$, let $\alpha_n(s)=\gamma_{n+m_0}(s)\gamma_{n-1+m_0}(s)\ldots \gamma_{1+m_0}(s)$. For any $s$, this continuous one-parameter family of automorphisms increases the radius of support of a local observable by at most $20n(n+m_0)$. It is generated by a continuous family of finite-range derivations $\derF_n(s)$ which for $n>1$ satisfy
\begin{equation}
\derF_{n}(s)=\derF_{n-1}(s)+\alpha_{n-1}(s)^{-1}(\derG_{n}(s)).
\end{equation}
Let us restrict $s$ to the interval $[0,1]$ and regard $\derF_n$ as an element of the \Frechet\ space $C^0([0,1],\mfkDal)$. Then 
\begin{equation}
\|\derF_n-\derF_{n-1}\|_\alpha=O\left((1+20n(n+m_0)\right)^\alpha g(n)),
\end{equation}
and since $g(r)=\Or$, the sequence $\{\derF_n\}$ is a Cauchy sequence and converges to a continuous family of almost local derivations $\derF(s)$ defined on $s\in [0,1]$. 

Let $\alpha(s)$ be the unique solution of (\ref{eq:LGAdef}) with the initial condition $\alpha(0)=\Id.$ By Lemma E.4 of \cite{LocalNoether}, for any $\CA\in\SAal$ and any $s\in [0,1]$, we have $\alpha(s)(\CA)=\lim_{n\ra\infty}\alpha_n(s)(\CA)$. On the other hand, $\lim_{n\ra\infty}\alpha_n(1)(\CA)=\beta\beta_{m_0}^{-1}(\CA).$ Hence, $\beta=\alpha(1)\beta_{m_0}$. Since $\beta_{m_0}$ is a zero-index QCA, it is an LGA. Therefore, $\beta$ is also an LGA. 
\end{proof}

We say that $\alpha \in \alLPA$ is strictly localized on a region $X$ if it is a trivial  extension of an automorphism of the algebra $\SAal_X$. We denote the subgroup of automorphisms $\alpha \in \alLPA$ strictly localized on $\RRp$ and $\RRm$ by $\alLPAsp$ and $\alLPAsm$, respectively. By Proposition \ref{prop:GNVWforLGA}, every 
automorphism in these subgroups has a  trivial GNVW index and can be generated by a continuous family of derivations $\derF: [0,1] \to \mfkDalsp$ or $\derF: [0,1] \to \mfkDalsm$. By $\alLPAz$ we denote the group of automorphisms of the form $\Ad_{\CU}$ for some unitary $\CU \in \SAal$. The GNVW index of any element of $\alLPAz$ is also trivial. The analog of Lemma \ref{lma:alphadecomposition} is the following lemma:

\begin{lemma} \label{lma:alphadecompositionsmooth}
Any $\alpha \in \alLPA$ with a trivial GNVW index admits a decomposition $\alpha = \alpha_{<0} \alpha_0\alpha_{\geq 0}$ for some $\alpha_{<0} \in \alLPAsm$, $\alpha_{0} \in \alLPAz$, $\alpha_{\geq 0} \in \alLPAsp$.
\end{lemma}

\begin{proof}
Let $\derF \in C^0([0,1],\mfkDal)$ be a family of derivations generating $\alpha$. We decompose $\derF=\derF_{<0}+\ad_{\chf_0}+\derF_{\geq 0}$ and let $\alpha_{\geq 0}$, $\alpha_{<0}$ be the LGAs generated by $\derF_{<0}$, $\derF_{\geq 0}$, respectively. The automorphism $\alpha\alpha^{-1}_{<0} \alpha^{-1}_{\geq 0}$ is generated by $(\alpha_{\geq 0} \alpha_{<0}) (\derF - \derF_{\geq 0} - \derF_{<0})\in C^0([0,1],\mfkDal)$. The latter family of derivations has the form $\ad_{\chh}$ where $\chh=(\alpha_{\geq 0} \alpha_{<0}) (\chf_0)\in C^0([0,1],\mfkdal)$. By Lemma \ref{lma:inner},   $\alpha \alpha^{-1}_{<0} \alpha^{-1}_{\geq 0}=\Ad_\CU$ for $\CU\in\SUal$. We can take $\alpha_0 = \Ad_{\alpha_{<0}^{-1}(\CU)}$.
\end{proof}

The decomposition from the lemma is not unique, and the next lemma characterizes this non-uniqueness.
\begin{lemma}\label{lma:nonuniquenesssmooth}
Suppose the GNVW index of $\alpha\in\alLPA$ is trivial. For any two decompositions $\alpha = \alpha_{<0} \alpha_0 \alpha_{\geq 0} = \tilde{\alpha}_{<0}  \tilde{\alpha}_0 \tilde{\alpha}_{\geq 0}$ as in Lemma \ref{lma:alphadecompositionsmooth}, we must have $\tilde{\alpha}^{-1}_{<0}  \alpha_{< 0} \in \alLPAz$ and $\tilde{\alpha}^{-1}_{\geq 0} \alpha_{\geq 0} \in \alLPAz$. 
\end{lemma}
\begin{proof}
Let $\beta_{<0}=\tilde\alpha_{<0}\alpha_{<0}^{-1}$ and $\beta_{\geq 0}=\tilde\alpha_{\geq 0}\alpha_{\geq 0}^{-1}$. Then $\beta_{<0} \in \alLPAsm$, $\beta_{\geq 0} \in \alLPAsp$, and $\beta_{<0} \beta_{\geq 0} = \Ad_{\CU}$ for some $\CU \in \SUal$.
Thus $\Ad_\CU$ induces an automorphism of $\SAql_{<0}$ as well as an automorphism of $\SAql_{\geq 0}$. As in the proof of Lemma \ref{lma:nonuniqueness}, this implies  $\CU=\CU_{<0}\CU_{\geq 0}$ for some $\CU_{<0}\in\SUql_{<0}$ and $\CU_{\geq 0}\in\SUql_{\geq 0}$. 

It remains to show that in fact  $\CU_{<0}\in\SUal_{<0}$ and $\CU_{\geq 0}\in\SUal_{\geq 0}$. 
For any $\CA\in\SAql$ and any $X\subset\RR$, let $E_X(\CA)\in\SAql_{X^c}$ be the conditional expectation value of $\CA$ in the tracial state on $\SAql_X$. Thus for any $\CA\in\SAql_{X^c}$ and any $\CB\in\SAql$ we have $E_X(\CA\CB)=\CA E_X(\CB)$ and $E_X(\CB\CA)=E_X(\CB)\CA$. The map $E_X:\SAql\ra\SAql_{X^c}$ satisfies $E_X\circ E_Y=E_{X\cup Y}$ for all $X,Y\subset\RR$. Also, $\|E_X(\CA)\|\leq \|\CA\|$ for any $\CA\in\SAql$.  For any $r\geq 0$, let $R(r)=(r,+\infty)$ and $L(r)=(-\infty,-r)$. It follows from Prop. D.1 of Ref. \cite{LocalNoether} that for any $\CA\in\SAal$ the function $g(r)=\|\CA-E_{R(r)}(\CA)\|$ is $\Or$. Therefore, 
\beq
\|\CU_{\geq 0}-E_{R(r)}\left(\CU_{\geq 0}\right)\|=\|\CU_{<0}\left(\CU_{\geq 0}-E_{R(r)}\left(\CU_{\geq 0}
\right)\right)\|=\|\CU-E_{R(r)}(\CU)\|=\Or.
\eeq
On the other hand,
\begin{multline}
\inf_{\CB\in B_0(r)}\|\CU_{\geq 0}-\CB\|\leq \|\CU_{\geq 0}-E_{L(r)\cup R(r)}\left(\CU_{\geq 0}\right)\|=\|E_{L(r)}(\CU_{\geq 0}-E_{R(r)}\left(\CU_{\geq 0}\right))\|\\
=\|\CU_{\geq 0}-E_{R(r)}\left(\CU_{\geq 0}\right)\|.
\end{multline}
Thus, $\inf_{\CB\in B_0(r)}\|\CU_{\geq 0}-\CB\|=\Or$ and  $\CU_{\geq 0}\in\SUal$. Then $\CU_{<0}=\CU\CU_{\geq 0}^*\in\SUal$.
\end{proof}

We denote the subgroup of $\alLPA$ generated by $\alLPAsm$ and $\alLPAz$ (resp. $\alLPAsp$ and $\alLPAz$) by $\alLPAm$ (resp. $\alLPAp$). Any element in $\alLPAm$ can be written as a product $\alpha_{<0}\alpha_0$ where $\alpha_{<0}\in\alLPAsm$ and $\alpha_0\in\alLPAz$. 
Any element in $\alLPAp$ can be written as a product $\alpha_0\alpha_{\geq 0}$ where $\alpha_{\geq 0}\in\alLPAsp$ and $\alpha_0\in\alLPAz$. The subgroups $\alLPAp$ and $\alLPAm$ are normal and their intersection is $\alLPAz$. This follows from Lemma \ref{lma:alphadecompositionsmooth} and Lemma \ref{lma:nonuniquenesssmooth}, respectively (cf. Section 2.2).

\subsection{Lie group symmetry of a quantum spin chain}

Let $M$ be a smooth manifold. A map $\alpha: M \to \alLPA$ is called smooth if for any $\CA \in \SAal$, the map $\alpha(\CA): M \to \SAal$ is smooth and there exists $\derG \in \Omega^1(M,\mfkDal)$ such that $d \alpha(\CA) = \alpha (\derG(\CA))$. It is clear from the last condition that $\derG$ is unique for a given smooth map $\alpha$. We denote it by $\alpha^{-1} d \alpha$.

\begin{remark}
Here and below, if $\alpha:M\ra\alLPA$ is smooth map, then $\alpha^{-1}:M\ra\alLPA$ will denote the smooth map which is the point-wise inverse of $\alpha$:  $\alpha^{-1}(p)=\alpha(p)^{-1}$, $\forall p\in M$. Similarly, given a pair of smooth maps $\alpha,\beta:M\ra\alLPA$, we denote by $\alpha\beta$ the map defined by $(\alpha\beta)(p)=\alpha(p)\beta(p)$. It is easy to check that both $\alpha^{-1}$ and $\alpha\beta$ are smooth if $\alpha$ and $\beta$ are smooth.
\end{remark}
\begin{remark}\label{rem:MCintegration}
    Let $\alpha:M\ra \alLPA$ be a smooth map and $\gamma:[0,1]\ra M$ be a smooth path. Then the one-parameter family of automorphisms $\alpha(\gamma(s))$ satisfies Eq. (\ref{eq:LGAdef}) with $\derF=\gamma^*\derG$ where $\derG=\alpha^{-1}d\alpha$. Therefore $\alpha(\gamma(1))=\alpha(\gamma(0))\alpha_{\gamma^*\derG}(1)$. 
\end{remark}
\begin{remark}\label{rem:integratingcontractibleW}
    Let $W$ be a smoothly contractible manifold. For any $\derG\in\Omega^1(W,\mfkDal)$ one can construct a smooth map $M\ra \alLPA$ as follows. Let $\gamma=\RR\times W\ra W$ be a smooth contracting homotopy (i.e. a smooth map such that for any $s\leq 0$ the restriction to $\{s\}\times W\ra W$ is a constant map to $w_0\in W$ and for any $s\geq 1$ the restriction to $\{s\}\times W\ra W$ is the identity map). We can regard $\gamma^*\derG$ as an element of $C^\infty([0,1],\mfkDal)$ smoothly depending on $w\in W$. By Remark E.1 in Ref. \cite{LocalNoether}, the family $\alpha_{\gamma^*\derG}(1)(w)$ over $W$ is smooth and satisfies $\alpha_{\gamma^*\derG}(1)(w_0)=\Id$.
\end{remark}

For a Lie group $G$, we say that a homomorphism $\alpha: G \to \alLPA$ is smooth if $\alpha$ is smooth as a map. A smooth homomorphism $\alpha$ induces a smooth homomorphism of Lie algebras $\mfkg \to \mfkDal$ given by $v\mapsto \alpha^{-1}d\alpha\vert_e(v),$ where $v\in T_e G$ and $e$ is the unit of $G$.

We can similarly consider smooth maps from a smooth manifold $M$ to the subgroups $\alLPAsm$, $\alLPAsp$. By definition, $\alpha: M \to \alLPA_{\pm}$ is smooth, if it is smooth as a map to $\alLPA$.

If $\derF\in C^0([0,1],\mfkDal)$ smoothly depends on parameters in $M$, then by Remark E.1 in \cite{LocalNoether} the family of LGAs $\alpha=\alpha_\derF(1):M\ra \alLPA$ is smooth. Together with Remark \ref{rem:inner}, this implies the following refinement of Lemma \ref{lma:alphadecompositionsmooth}:

\begin{lemma} \label{lma:alphadecompositionsmooth2}
Let $W\subset\RR^n$ be a smoothly contractible open subset. Any smooth family $\alpha: W \to \alLPA$ with a trivial GNVW index admits a decomposition $\alpha = \alpha_{<0} \alpha_0\alpha_{\geq 0}$ for some smooth $\alpha_{<0}: W \to \alLPAsm$, $\alpha_{0}: W \to \alLPAz$, $\alpha_{\geq 0}: W \to \alLPAsp$. Furthermore, one can choose a smooth $\CU:W\ra \SUal$ such that $\alpha_0=\Ad_\CU$.
\end{lemma}
\begin{proof}
    Let $\gamma=\RR\times W\ra W$ be a smooth contracting homotopy. By Lemma \ref{lma:alphadecompositionsmooth}, we can write $\alpha(w_0)=\beta_{<0}\beta_0\beta_{\geq 0}$ for some $\beta_{<0}\in\alLPAsm$, $\beta_{\geq 0}\in\alLPAsp$ and $\beta_0\in\alLPAz$. Let $\derG=\alpha^{-1}d\alpha$. By Remark \ref{rem:integratingcontractibleW}, the families of LGAs $\alpha_{\gamma^*\derG}(1)(w)$, $\alpha_{\gamma^*\derG_{<0}}(1)(w)$, 
    $\alpha_{\gamma^*\derG_{\geq 0}}(1)(w)$ over $W$ are smooth. Further, by Remark \ref{rem:MCintegration}, for all $w\in W$ we have $\alpha(w)=\alpha(w_0)\alpha_{\gamma^*\derG}(1)(w).$

    Let $\alpha_{<0}(w)=\beta_{<0}\alpha_{\gamma^*\derG_{<0}}(1)(w)$. This is a smooth map from $W$ to $\alLPAsm$. Similarly, we define a smooth map $\alpha_{\geq 0}:W\ra \alLPAsp$ by
    $\alpha_{\geq 0}(w)=\beta_{\geq 0}\alpha_{\gamma^*\derG_{\geq 0}}(1)(w)$. Then using Remark \ref{rem:inner} it is easy to check that $\alpha\alpha_{<0}^{-1}\alpha_{\geq 0}^{-1}$ is a smooth map $W\ra\alLPAz$ of the form $\Ad_\CV$, where $\CV:W\ra \SUal$ is smooth. Therefore we can set $\alpha_0=\Ad_\CU$ where $\CU=\alpha_{<0}^{-1}(\CV).$
    
\end{proof}
The decomposition in the above lemma is not unique. The next lemma characterizes this non-uniqueness.
\begin{lemma}\label{lma:nonuniquenesssmooth2}
Let $W\subset\RR^n$ be a smoothly contractible open subset and $\alpha:W\ra\alLPA$ be a smooth map whose image consists of automorphisms with a vanishing GNVW index. For any two decompositions as in Lemma \ref{lma:alphadecompositionsmooth2}, $\alpha = \alpha_{<0} \alpha_0 \alpha_{\geq 0} = {\tilde \alpha}_{<0} {\tilde \alpha}_0 {\tilde \alpha}_{\geq 0}$, we have  ${\tilde\alpha}_{<0}^{-1}\alpha_{<0}=\Ad_{\CV}$
for some smooth map $\CV:W\ra \SUal$. 
\end{lemma}
\begin{proof}
By the argument in Lemma \ref{lma:nonuniquenesssmooth}, it is sufficient to show that any smooth $\CU:W\ra\SUal$ such that $\Ad_\CU$ induces automorphisms of both $\SAal_{<0}$ and $\SAal_{\geq 0}$ can be written as $\CU_{<0}\CU_{\geq 0}$, where $\CU_{<0}:W\ra\SUal_{<0}$ and $\CU_{\geq 0}:W\ra\SUal_{\geq 0}$ are smooth. Pick $w_0\in W$ and let $\gamma:\RR\times W\ra W$ be a  contracting homotopy with $\gamma(0,w)=w_0$ for all $w\in W$. By the argument in Lemma \ref{lma:nonuniquenesssmooth}, there exist $\CV_{<0}\in \SUal_{<0}$ and $\CV_{\geq 0}\in \SUal_{\geq 0}$ such that $\CU(w_0)=\CV_{<0}\CV_{\geq 0}$. 

Let $\chh=\CU^{-1}d\CU\in \Omega^1(W,\mfkdal)$. For any $\CA_{<0}\in\SAl_{<0}$, any $\CA_{\geq 0}\in\SAl_{\geq 0}$, and any $w\in W$ we have $[\chh(w),\CA_{<0}]\in\SAal_{<0}$ and $[\chh(w),\CA_{\geq 0}]\in\SAal_{\geq 0}$. Corollary 3.3 from \cite{Lance} implies that $\chh(w)$ is a sum of an observable from $\SAql_{<0}$ and an observable from $\SAql_{\geq 0}$. By writing $\chh^Y = (\ad_{\chh})^Y$, we see that $\chh^Y\neq 0$ only if $Y\subset (-\infty,0)$ or $Y\subset [0,+\infty)$, so $\chh$ uniquely decomposes as a sum $\chh_{<0}+\chh_{\geq 0}$ where
$\chh_{<0}=\sum_{Y\subset (-\infty,0)} \chh^Y\in\Omega^1(W,\mfkdalsn)$ and $\chh_{\geq 0}=\sum_{Y\subset [0,+\infty)}\chh^Y\in\Omega^1(W,\mfkdalsp)$. Here $\mfkdalsn$ (resp. $\mfkdalsp$) is the subspace of elements of $\SAal_{<0}$ (resp. $\SAal_{\geq 0}$) that belong to $\mfkdal$.

Let $\tilde\CU=\gamma^*\CU\in C^\infty([0,1]\times W, \SUal)$, so that ${\tilde\CU}^{-1}d\tilde\CU=\gamma^*\chh=\gamma^*\chh_{<0}+\gamma^*\chh_{\geq 0}$. Let ${\tilde\CU}_{<0}:C^\infty([0,1]\times W,\SUal_{<0})$ be the solution of eq. \ref{eq:LGAinnerdef} with the initial condition ${\tilde\CU}_{<0}(0,w)=\CV_{<0}$. Let ${\tilde\CU}_{\geq 0}:C^\infty([0,1]\times W,\SUal_{\geq 0})$ be the solution of eq. \ref{eq:LGAinnerdef} with the initial condition ${\tilde\CU}_{\geq 0}(0,w)=\CV_{\geq 0}$. Then $\tilde\CU_{<0}\tilde\CU_{\geq 0}$ satisfies the same differential equation with respect to $s\in [0,1]$ and the same initial condition at $s=0$ as $\tilde\CU$, and thus by Lemma \ref{lma:inner}, $\tilde\CU=\tilde\CU_{<0}\tilde\CU_{\geq 0}$. Letting $s=1$, we get the desired statement.
\end{proof}

\begin{corollary}\label{cor:intersectionsmooth}
Let $W$ and $\alpha$ be as in Lemma \ref{lma:nonuniquenesssmooth2}, and suppose the image of $\alpha$ is in $\alLPAp\cap\alLPAm$. Then there exists a smooth map $\CW:W\ra \SUal$ such that $\alpha=\Ad_\CW$.
\end{corollary}
\begin{proof}
By Lemma \ref{lma:alphadecompositionsmooth2}, we have $\alpha=\Ad_\CU\alpha_{\geq 0}={\tilde \alpha}_{<0}\Ad_{\tilde\CU}$ for some smooth maps $\tilde\alpha_{<0}:W\ra\alLPAm$, $\alpha_{\geq 0}:W\ra\LPAp$, and $\CU,\tilde\CU:W\ra\SUal$. Then by Lemma \ref{lma:nonuniquenesssmooth2} we have $\alpha_{\geq 0}=\Ad_\CV$ and $\alpha_{<0}=\Ad_{\tilde\CV}$ for some smooth maps $\CV:W\ra\SUal$ and $\tilde\CV:W\ra\SUal$. Therefore  $\alpha=\Ad_{\CU\CV}$. 
\end{proof}

\subsection{Definition of the anomaly index}

Here we define an $H^3_{diff}(G,U(1))$-valued anomaly index for a smooth homomorphism from a Lie group $G$ to the group of almost local LPAs 
$\alLPA$.

\subsubsection{Anomalous Lie symmetries in quantum mechanics}

Let us first discuss the case of quantum mechanics. If $\CH$ is the Hilbert space of a quantum mechanical system, then a Lie group symmetry corresponds to a smooth homomorphism $\alpha: G \to PU(\CH)$ from $G$ to the projective unitary group $PU(\CH)$ of the Hilbert space $\CH$. By Bargmann's theorem \cite{Bargmann}, it is always possible to lift it to a smooth homomorphism $\tilde{G} \to U(\CH)$, where $\tilde{G}$ is a Lie group which is a central extension of $G$ by $U(1)$ and simultaneously a locally trivial  circle bundle over $G$, but it might not be possible to lift $\alpha$ to a homomorphism $G \to U(\CH)$. The obstruction is given by an element of $H_{diff}^2(G, U(1))$. 

A cocycle that represents this element can be defined as follows. Let us choose a good simplicial open cover $U_{\bullet,\bullet}$ of $B_{\bullet} G$. On $U_{1,0}$, we define a smooth family of unitary observables $\CV:G \to U(\CH)$ which lift the homomorphism, i.e. $\Ad_{\CV} = \alpha$. This can always be done as all the charts of the cover are contractible. On $U_{2,0}$ and $U_{1,1}$, the expressions $(d^*_2 \CV)(d^*_0 \CV)(d^*_1 \CV)^{-1}$ and $(\delta^*_1 \CV)(\delta^*_0 \CV)^{-1}$, respectively, define $U(1)$-valued smooth functions. It is easy to see that they satisfy the 2-cocycle condition of the double complex and therefore define an element of  $H^2_{diff}(G,U(1))$. This cohomology class is independent of the choice of $\CV$. Indeed, any two choices of $\CV$ are related by multiplication by a smooth $U(1)$-valued function. Therefore, $(d^*_2 \CV)(d^*_0 \CV)(d^*_1 \CV)^{-1}$ and $(\delta^*_1 \CV)(\delta^*_0 \CV)^{-1}$ can change at most by a coboundary.

\subsubsection{Anomalous Lie symmetries of quantum spin chains}

Let us now define an $H^3_{diff}(G,U(1))$-valued index for a smooth homomorphism $\alpha:G \to \alLPA$.

First, we consider the case when the image of $\alpha$ consists of almost local LPAs which have a trivial GNVW index. Let $U_{\bullet,\bullet}$ be a good simplicial cover of $B_{\bullet} G$. We have a map $\iota:U_{1,0} \to G$. By Lemma \ref{lma:alphadecompositionsmooth2}, we can choose a smooth map $\beta:U_{1,0} \to \alLPAp$ such that $(\iota^*\alpha) \beta^{-1}$ defines a smooth map $U_{1,0} \to \alLPAm$. Since $(d^*_2\iota^* \alpha)(d^*_0 \iota^* \alpha) (d^*_1 \iota^* \alpha)^{-1}=\Id$ and $(\delta^*_{1} \iota^* \alpha) (\delta^*_{0} \iota^* \alpha)^{-1} = \Id$, by Corollary \ref{cor:intersectionsmooth} (cf. also Section 3.2), there exist smooth maps $\CV:U_{2,0} \to \SUal$ and $\CW:U_{1,1} \to \SUal$ such that $(d^*_2 \beta)(d^*_0 \beta) (d^*_1 \beta)^{-1} = \Ad_{\CV}$ and $(\delta^*_{1} \beta) (\delta^*_{0} \beta)^{-1} = \Ad_{\CW}$. The expressions 
\beq\label{eq:omegaLie}
\omega = (d^*_3 \CV) (d^*_1 \CV) (d^*_2 \CV)^{-1} ((d^*_3 d^*_2 \beta) (d^*_0 \CV))^{-1}
\eeq
on $U_{3,0}$, 
\beq\label{eq:upsLie}
\ups=(d^*_1 \CW)^{-1}(\delta^*_1 \CV)^{-1}((\delta^*_1 d^*_1 \beta)(d^*_0 \CW))(d^*_2 \CW)(\delta^*_0 \CV)
\eeq
on $U_{2,1}$ and 
\beq\label{eq:etaLie}
\eta=(\delta^*_2 \CW)(\delta^*_0 \CW)(\delta^*_1 \CW)^{-1}
\eeq
on $U_{1,2}$ define smooth $U(1)$-valued functions. In Appendix \ref{app:cocycle}, we show that  together they satisfy the cocycle condition and therefore define a cohomology class $[(\omega ,\ups, \eta)] \in H_{diff}^3(G,U(1))$.

\begin{prop}
The class $[(\omega,\ups,\eta)] \in H_{diff}^3(G,U(1))$ does not depend on the choice of $(U_{\bullet,\bullet},\beta, \CV, \CW)$.
\end{prop}

\begin{proof}
For fixed $(U_{\bullet,\bullet},\beta)$, any two choices of $\CV$ and $\CW$ differ by some smooth 
$U(1)$-valued functions on $U_{2,0}$ and $U_{1,1}$, respectively. Such functions manifestly modify $(\omega,\ups,\eta)$ by a coboundary. Therefore, the class $[(\omega,\ups,\eta)]$ can depend only on the choice of $\beta$ and $U_{\bullet,\bullet}$.

Let us make a choice of $(U_{\bullet,\bullet},\beta, \CV, \CW)$. Any other choice of $\beta$ is given by $\tilde{\beta} = \beta\Ad_{\CU}$ for some smooth map  $\CU:U_{1,0} \to \SUal$. We can take the corresponding $\CV$ and $\CW$ to be
\beq
\tilde{\CV} = (d^*_2\beta(\CU)) \CV (d^*_{1} \beta)((d^*_{0} \CU) (d^*_{1} \CU)^{-1}),\quad 
\tilde{\CW} = \CW (\delta^*_0 \beta)((\delta^*_{1} \CU)(\delta^*_{0} \CU)^{-1}).
\eeq
One can argue that the class $[(\omega,\ups,\eta)]$ is independent of $\CU$ by continuity, as the former is discrete, while the latter can be smoothly deformed to the identity thanks to the connectedness of the group $\SUal$. Instead, we show that the cocycle $(\omega,\ups,\eta)$ is the same for $(\beta,\CV,\CW)$ and $(\tilde{\beta}, \tilde{\CV}, \tilde{\CW})$ by a direct computation in Appendix \ref{app:IndependenceOfBetaL}. Thus, the class $[(\omega,\ups,\eta)]$ is independent of the choice of $\beta$.

Finally, independence of the choice of $U_{\bullet,\bullet}$ follows from the fact that the construction of the class $[(\omega,\ups,\eta)]$ is compatible with the refinement of simplicial covers.

\end{proof}

\begin{remark}
Instead of restrictions to the right half-chain, one can use the restrictions tomed the left one and similarly define a class in $H_{diff}^3(G,U(1))$. One can show that this class is the negative of the class defined above. We omit the details. 
\end{remark}

\begin{remark}\label{rem:discrete} If $G$ is discrete, then $\CW$, $\ups$, $\eta$ can be chosen to be the identity functions, and therefore $\omega$ defines a cocycle for the usual group cohomology $H^3(G,U(1))$.
\end{remark}

\begin{remark}
A similar procedure can be used to define an index for 2d SPT phases for a compact Lie group. Such indices were defined in \cite{sopenko2021,ogata2021h} for finite groups and in \cite{LocalNoether} for connected compact Lie groups. In order to extend the definition to arbitrary compact Lie groups, one can use the procedures similar to \cite{sopenko2021,ogata2021h} with the modification that all restricted automorphisms must be smooth functions on the charts of the cover $U_{1,0}$.  
\end{remark}

A generalization to the case when the image of $\alpha$ does not comprise only automorphisms with a trivial GNVW index proceeds in exactly the same way as for an abstract symmetry group $G$ and will not be repeated here.

\subsection{Triviality of the anomaly index for compact connected Lie groups}

First, we define sub-algebras of the \Frechet -Lie algebra $\mfkDal$ associated with arbitrary subsets $X\subset \RR$. For any two subsets $X,Y\subset\RR$ let $\dist(X,Y)$ be the distance between them, i.e. $\dist(X,Y)=\inf_{x\in X,y\in Y}|x-y|$.
\begin{definition}
 $\mfkDal_X$ is the space of functions $\derH:\Br_1\ra\mfkdl$ such that $\derH\in \mfkDal$ and for all $\alpha\in\NN$
\beq
\sup_{Y\in\Br_1} (1+\dist(Y,X))^\alpha \|\derH^Y\|<\infty . 
\eeq
We will use a short-hand $\mfkDal_{[0,+\infty)}=\mfkDalp$ and $\mfkDal_{(-\infty,0]}=\mfkDalm$.
\end{definition}

By definition, $\mfkDal_X$ is a subspace of $\mfkDal$, and it is easy to check that it is a sub-algebra. But it is not a closed sub-algebra. While there is a natural topology on $\mfkDal_X$ which makes it into a \Frechet -Lie algebra, we do not use it below. Note also that $\mfkDal_{\geq 0}\subset\mfkDal_+$, $\mfkDal_{<0}\subset \mfkDal_-$, and the image of the injection $\ad:\mfkdal\ra\mfkDal$ is precisely $\mfkDal_+\cap\mfkDalm$. 

\begin{lemma} \label{lma:FplusGenerates}
 Let $\derF\in\mfkDalp$ (resp. $\derF\in\mfkDalm$). Then the almost local automorphism $\alpha_\derF(1)$ is in $\alLPAp$ (resp. $\alLPAm$). Here we regard $\derF$ as a constant element of $C^0([0,1],\mfkDal)$. 
 \end{lemma}
\begin{proof}
Suppose $\derF \in \mfkDalp$. Let $\chh \in \mfkdal$ and $\derH \in \mfkDalsp$ be such that $\derF = \ad_{\chh} + \derH$. By Lemma \ref{lma:inner}, there is a unique solution $\CU \in C^1([0,1],\SUal)$ of the differential equation 
\beq
\frac{d} {ds} \CU(s) = \CU(s)\, \alpha_{\derH}(s)(\chh)
\eeq
with the initial condition $\CU(0) = 1$. By a direct computation, we have $\alpha_{\derF}(s) = \Ad_{\CU(s)}  \alpha_{\derH}(s)$. Since $\alpha_{\derH}(s) \in \LPAsp$, we have $\alpha_{\derF}(s) \in \alLPAp$. The case $\derF \in \mfkDalm$ can be treated similarly. 
\end{proof}

Let us show that the anomaly index of any $U(1)$ symmetry of a 1d quantum spin system is trivial\footnote{Since we only discuss spin systems which have a finite-dimensional on-site Hilbert space, our result does not rule out the possibility of an anomalous $U(1)$ action for systems with infinite-dimensional on-site Hilbert spaces or more complicated algebras of observables obtained by imposing local constraints.}. We parameterize the elements of the group $U(1)$ as $g = e^{i \theta}$, $\theta\in\RR/2\pi\ZZ$. Let $\alpha(\theta)$ be a smooth homomorphism from $U(1)$ to the group of LGAs and $\derQ \in \mfkDal$ be the derivation that generates it. 
\begin{lemma} \label{lma:Qaveraging}
   For any $\derF\in\mfkDal_{\geq 0}$ (resp. any $\derF\in\mfkDal_{<0}$), the derivation
   \beq\derG=\frac{1}{2\pi}\int_0^{2\pi}\alpha(\theta)(\derF) d\theta
   \eeq 
   belongs to $\mfkDal_+$ (resp. to $\mfkDal_-$) and is $U(1)$-invariant.
\end{lemma}
\begin{proof}
Suppose $\derF\in\mfkDalsp$. Let $\derQ_{<0} \in \mfkDalsm$, $\derQ_{\geq 0} \in \mfkDalsp$, $\chq \in \mfkdal$ be such that $\derQ = \derQ_{<0} + \ad_{\chq} + \derQ_{\geq 0}$. Let $\alpha_{<0}(\theta)$ be the smooth family of automorphisms generated by $\derQ_{<0}$. We have $\alpha_{<0}(\theta) \in \alLPAsm$. By Lemma \ref{lma:alphadecompositionsmooth2}, there are smooth families (over the interval $\theta \in [0, 2 \pi]$) $\alpha_{\geq 0}(\theta) \in \alLPAsp$ and $\CU(\theta) \in \SUal$ such that $\alpha(\theta) = \Ad_{\CU(\theta)} \alpha_{\geq 0}(\theta)  \alpha_{< 0}(\theta)$. We have $\alpha(\theta)(\derF) = \ad_{\chh(\theta)}+\derH(\theta)$, where $\derH(\theta) = \alpha_{\geq 0}(\theta)(\derF)$ defines a smooth map $[0,2\pi]\ra\mfkDalsp$ and $\chh(\theta)=\CU(\theta)\derH(\theta)(\CU(\theta)^*)$ defines a smooth map $[0,2\pi]\ra\mfkdal$. Since $\mfkDalsp$ and $\ad(\mfkdal)$ are subspaces of $\mfkDalp$, we conclude that $\derG \in \mfkDalp$. The case  $\derF\in\mfkDalsm$ is treated similarly. 

$U(1)$-invariance of $\derG$ follows from
\beq
[\derQ,\derG]=\frac{1}{2\pi}\int_0^{2\pi}\alpha(\theta)([\derQ,\derF])d\theta = \frac{1}{2\pi}\int_0^{2\pi} \frac{d}{d\theta}\alpha(\theta)(\derF) d\theta = \alpha(2 \pi)(\derF) - \derF = 0.
\eeq
\end{proof}

We are now ready to prove
\begin{prop}\label{prop:trivialanomalyindex}
    Let $\alpha(\theta)$, $\theta\in \RR/2\pi\ZZ$ be a smooth action of $U(1)$ on a 1d spin system. Then the corresponding anomaly index taking values in $H^3_{diff}(U(1),U(1))\simeq H^4(BU(1),\ZZ)\simeq\ZZ$ vanishes. 
\end{prop}

\begin{proof}
Recall that $\mfkDal$ is a closed sub-space of $\mfkDal_{\geq 0}\oplus\mfkDal_{<0}\oplus\mfkdal$. Let $\derQ_{\geq 0}$ be the component of $\derQ$ in $\mfkDal_{\geq 0}$, and let 
\beq
\derQ_+=\frac{1}{2\pi}\int_0^{2\pi}\alpha(\theta)(\derQ_{\geq 0}) d\theta
\eeq
By Lemma \ref{lma:Qaveraging}, $\derQ_+\in\mfkDal_+$ and is $U(1)$-invariant. Further, if we let $\derQ_-=\derQ-\derQ_+$, then we have 
\beq
\derQ_-=\frac{1}{2\pi}\int_0^{2\pi} \alpha(\theta)(\derQ-\derQ_+)d\theta,
\eeq
and thus  $\derQ_-\in\mfkDal_-$. 

Let $U_{\bullet,\bullet}$ be a good simplicial cover of $B_\bullet U(1)$. Each chart of $B_1U(1)=U(1)$ is an interval that either contain $g = 1$ or does not. If it does contain it, we let $\beta\left(e^{i\theta}\right)=\alpha_{\derQ_+}(\theta)$ for $\theta \in (0 ,2 \pi)$. Otherwise, we let $\beta\left(e^{i\theta}\right)=\alpha_{\derQ_+}(\theta)$ for $\theta \in (-2\pi + \theta_c, \theta_c)$, where $\theta_c \in (0 , 2\pi)$ is any point that does not belong to the chart. Since $[\derQ_+,\derQ_-]=0$ and $\alpha_\derQ (2 \pi)=\Id,$ we have
$\alpha_{\derQ_+}(2 \pi)\alpha_{\derQ_-}(2 \pi)=\Id$. Therefore by Lemma \ref{lma:nonuniquenesssmooth}, there exists a $\CU\in\SUal$ such that $\Ad_{\CU} = \alpha_{\derQ_+}(2 \pi)$.
With this choice of $\beta(g)$, both $\CV$ and $\CW$ can be chosen to be locally constant and equal to either 1 or $\CU$ or $\CU^{-1}$. 

Next we prove that $\derQ_+(\CU)=0$. Let $L \in \NN$, and let $\derQ_{+,L}$ be the derivation defined by $\derQ_{+,L}^Y=0$ for any $Y\in \Br_1$ such that $Y\subset [L+1/2,+\infty)$ and $\derQ_{+,L}^Y=\derQ_+^Y$ for any other $Y\in\Br_1$.  We have $\derQ_{+,L} = \ad_{\chq_L}$ where $\chq_L \in \mfkdal$ is given by 
\beq
\chq_L=\sum_{Y\cap (-\infty,L]\neq\emptyset} \derQ_+^Y.
\eeq
This sum is \Frechet -convergent because (by Lemma \ref{lma:Qaveraging}) $\derQ_+\in\mfkDalp$. Also, it is easy to see that $[\derQ_-,\derQ_+-\derQ_{+,L}]$ has the form $\ad_\chf$ for some $\chf\in\mfkdal$ with $\|\chf\|=\OL$. Since $[\derQ_-,\derQ_+]=0$, this implies $\derQ_-(\chq_L) = \OL$, and therefore 
\beq
[\CU,\chq_L]\CU^* = \CU\chq_L\CU^* - \chq_L = \alpha_{\derQ_-}(-2\pi)(\chq_L) - \chq_L = - \int_{0}^{2 \pi} d \theta \alpha_{\derQ_-}(-\theta)(\derQ_-(\chq_L)) = \OL .
\eeq
On the other hand, since $\derQ_+(\CU) = [\chq_L,\CU] + \OL$, we conclude that $\derQ_+ (\CU) = \OL$. Since $\derQ_+ (\CU)$ is independent of $L$ and $L$ is arbitrary, we get $\derQ_+ (\CU) = 0$.

This implies that $[(\omega,\ups,\eta)]$ is trivial. Indeed, we have $\omega=\CU^{d\mu}$, $\ups=\CU^{\delta \mu+d\nu}$,$\eta=\CU^{\delta \nu}$ for some $(\mu,\nu)\in C^2(U_{\bullet,\bullet},\ZZ)$. If the 3-cochain $(d\mu, \delta \mu+d \nu,\delta\nu)$ is zero, we are done. If it is not zero, let $N$ be the least common multiple of the set of all nonzero vales of its components. Since the values of the cochain $(\omega,\ups,\eta)$ are proportional to the identity observable, the same must hold for $\CU^N$. Therefore the cohomology class  $[(\omega,\ups,\eta)]$ is annihilated by $N$. But $H^3_{diff}(U(1),U(1))\simeq H^4_{sing}(BU(1),\ZZ)\simeq\ZZ$ is torsion-free, so the cohomology class $[(\omega,\ups,\eta)]$ is trivial.

\end{proof}

\begin{corollary}
    For any compact connected Lie group $G$ and any smooth homomorphism $\alpha:G\ra\alLPA$, the anomaly index vanishes.
\end{corollary}
\begin{proof}
First, the anomaly index vanishes when $G$ is any compact connected abelian Lie group, i.e. when $G$ is a torus $\TT$. Indeed, on the one hand, the anomaly index is functorial in $G$, so by Prop. \ref{prop:trivialanomalyindex} must vanish when restricted to any $U(1)$ subgroup of $\TT$. On the other hand, since $H^3_{sing}(B\TT,U(1))$ is trivial, the coefficient long exact sequence implies that $H^3_{diff}(B\TT,U(1))\simeq H^4_{sing}(B\TT,\ZZ)$ injects into $H^4_{sing}(B\TT,\RR)$. The latter group is a real vector space which is canonically identified with the space of homogeneous  
quadratic polynomials on the Lie algebra of $\TT$. Thus an element of $H^4_{sing}(B\TT,\ZZ)$ is trivial if its restriction to all $U(1)$ subgroups of $\TT$ is trivial.

This implies that the anomaly index vanishes for any compact connected Lie group $G$. Indeed, one can show that for such a $G$, the inclusion of a maximal torus $\TT$ into $G$ induces an isomorphism of $H^3_{diff}(G,U(1))\simeq H^4_{sing}(BG,\ZZ)$ with the Weyl-invariant part of $H^4_{sing}(B\TT,\ZZ)$ \cite{Henriques}.
Thus the anomaly index for $G$ is determined by the anomaly index for any maximal torus of $G$, and we have shown that the latter is trivial. 
\end{proof}

\subsection{Mixed anomaly between translations and an on-site Lie symmetry}

Let us choose a smooth projective representation $\bpi$ of a compact Lie group $G_0$ on a finite-dimensional Hilbert space $\HilbV$. It defines a class in $H^2_{diff}(G_0,U(1))\simeq H^2_{sing}(BG_0,U(1))$ (see Section A.3.1). One can describe this class explicitly as follows. One picks a good simplicial cover $U_{\bullet,\bullet}$ of the simplicial manifold $B_\bullet G_0$. On $U_{1,0}$, one can lift $\bpi$ to a smooth map $G_0\ra \UnitaryOper(\HilbV)$ which we also denote by $\bpi$. Let $\rho=(d_0^*\bpi)^{-1} (d_1^*\bpi)(d_2^*\bpi)^{-1}\in C^\infty(U_{2,0},U(1))$ and $\sigma=(\delta_0^*\bpi )^{-1}\left(\delta_1^*\bpi\right)\in C^\infty(U_{1,1},U(1))$. Then $(\rho,\sigma)$ is a 2-cocycle and $[(\rho,\sigma)]\in H^2_{diff}(G_0,U(1))$ is the desired class.

Consider a 1d spin system which has $\HilbV$ as the on-site Hilbert space and a homomorphism from $G=G_0\times\ZZ$ to $\alLPA$ such that $(g, n) \in G_0 \times \ZZ$ is mapped to 
\beq
\tau^n  \prod_{j\in\ZZ}\Ad_{\bpi_j(g)},
\eeq
where $\tau$ is the translation to the right by one site. 
Let us compute the corresponding anomaly class in $H^3_{diff}(G,U(1))$. It vanishes when pulled back to $H^3_{diff}(G_0,U(1))$ since the action of $G_0$ is on-site. Since $H^3_{diff}(G,U(1))=H^3_{diff}(G_0,U(1))\oplus H^2_{diff}(G_0,U(1))$ (see Section A.1), the anomaly index for $G_0\times\ZZ$ can be regarded as an element of  $H^2_{diff}(G_0,U(1))$. 

We claim that the anomaly index is equal to $[(\rho,\sigma)]$. To see this, we follow the same steps as in Example  \ref{example:LSMfinite}: stack the system with its copy on which only $\ZZ$ acts non-trivially, so that the action of $G_0\times\ZZ$ on the composite is by LGAs. We use the same restriction of the automorphisms $\alpha(g,n)$ on the right half-chain as in Section \ref{example:LSMfinite}, namely (\ref{eq:betaLSM}). These automorphisms can be defined globally on $G$ because the action of $G_0$ is on-site. As a result $\CW:U_{1,1}\times \ZZ\ra \SUal$ can be chosen  trivial, while $\CV:U_{2,0}\ra\SUal$ is given by the same formula as before:
\beq
\CV(g,0;g',0)=\CV(g,1;g',0)=1,\quad \CV(g,0;g',1)=\bpi_1(g)\otimes 1.
\eeq
The only difference is that now $\CV$ cannot be defined globally on $B_2G$. Thus, $\eta$ defined by (\ref{eq:etaLie}) is trivial, but both $\omega$ and $\ups$ are non-trivial. 

It remains to evaluate the 3-cochain $(\omega,\ups,\eta)\in C^3(U_{\bullet,\bullet},U(1))$ and to compute its slant product with the generator $[1]$ of $H_1(\ZZ,\ZZ)$. This gives $[(\omega,\ups,\eta)]/[1]=[(\rho,\sigma)]$.

\subsection{Absence of symmetric gapped states}

To prove a version of Theorem \ref{thm:gLSM} for Lie group symmetries realized by $\alLPA$, the split property of gapped 1d states is not enough. One also needs to use a certain property of derivations which do not excite a gapped state.
\begin{definition}
    A derivation $\derF\in\mfkDal$ does not excite a state $\psi$ if $\psi(\derF(\CA))=0$ for all $\CA\in\SAal$. Similarly, $\derF\in \Omega^1(M,\mfkDal)$, where $M$ is a manifold, does not excite $\psi$ if $\psi(\derF(\CA))=0$ for all $\CA\in\SAal$.
\end{definition}
\begin{remark}\label{rem:donotexcitepreserve}
If $\derF\in C^0([0,1],\mfkDal)$ is a family of derivations each of which does not excite $\psi$, then $\alpha_\derF(1)\in\alLPA$ preserves $\psi$.
\end{remark}
\begin{remark}\label{rem:derGNSideal}
    If $\derF\in\mfkDal$ is an anti-self-adjoint derivation which does not excite $\psi$, then $\derF$ preserves the intersection of the GNS ideal of $\psi$ and $\SAal$. Indeed, suppose $\psi(\CA^*\CA)=0$ for some $\CA\in\SAal$, then 
\beq
\psi(\derF(\CA)^*\derF(\CA))=\psi(\derF(\CA^*\derF(\CA)))-\psi(\CA^*\derF(\derF(\CA)))=0.
\eeq
Therefore, if we let $\CD_\psi\subset \CH_\psi$ be a dense subspace spanned by vectors of the form $\pi_\psi(\CA)|\psi\ral$, $\CA\in\SAal$, then to every such $\derF$ one can associate an unbounded operator $F:\CD_\psi\ra\CH_\psi$ by letting
\beq
F\pi_\psi(\CA)|\psi\ral=\pi_\psi(\derF(\CA))|\psi\ral.
\eeq
The operator $F$ is skew-symmetric and satisfies $F|\psi\ral=0$.
\end{remark}

The following result is a corollary of Theorem 4 from Ref. \cite{LocalNoether}:
\begin{corollary}\label{cor:exactness}
    Let $\psi$ be a gapped state. Suppose $\derF\in \Omega^1(M,\mfkDal)$ does not excite $\psi$. Then there exist $\derF_{<0} \in \Omega^1(M,\mfkDalsm)$, $\chf_{\pm}\in\Omega^1(M, \mfkdal)$, $\derF_{\geq 0}\in\Omega^1(M,\mfkDalsp)$ such that
    $\derF=\derF_{<0} + \ad_{\chf_{-}} + \ad_{\chf_{+}} + \derF_{\geq 0}$ and $\derF_{<0} + \ad_{\chf_-}$ and $\derF_{\geq 0} + \ad_{\chf_+}$ do not excite  the state $\psi$.
\end{corollary}

\begin{theorem} \label{thm:gLSM2}
Let $G$ be a Lie group. Suppose $G$ acts on a 1d quantum spin chain by a smooth homomorphism $\alpha: G \to \alLPA$. Suppose also there exists a $G$-invariant finite-range Hamiltonian $\derH \in \mfkDl$ whose ground state  $\psi$ is $G$-invariant and gapped. Then the anomaly index $[(\omega, \ups, \eta)] \in H_{diff}^3(G,U(1))$ vanishes.
\end{theorem}

\begin{proof}
Suppose first that all automorphisms in the image of $\alpha$ have a trivial GNVW index. As in the definition of the anomaly index above, we choose a good simplicial open cover $U_{\bullet,\bullet}$ of $B_{\bullet} G$. The derivation $\alpha^{-1} d \alpha$ preserves the state $\psi$. By Corollary \ref{cor:exactness}, we can represent $\alpha^{-1} d \alpha$ as a sum $\derG_{<0} + \ad_{\chg_{-}} + \ad_{\chg_{+}} + \derG_{\geq 0}$ where $\derG_{<0} \in \Omega^1(G,\mfkDalsm)$, $\chg_{\pm}\in\Omega^1(G, \mfkdal)$, $\derG_{\geq 0}\in\Omega^1(M, \mfkDalsp)$ so that $\derG_{<0} + \ad{\chg_-}$ and $\derG_{\geq 0} + \ad_{\chg_+}$ do not excite  the state $\psi$. By Remark \ref{rem:donotexcitepreserve} and Lemma \ref{lma:unitaryeq}, there exist smooth families of LGAs $\beta:U_{1,0} \to \alLPAsp$ such that $\alpha\beta^{-1}:U_{1,0}\ra \alLPAsm$ and $\psi \circ \beta$ is constant on each connected component. 

We denote the components of a function $f$ on $U^{(p)}_{a_1} \cap ... \cap U^{(p)}_{a_q} \subset G^p$ at $(g_1,...,g_p) \in G^p$ by $f^{(g_1,...,g_p)}_{a_1...a_q}$. By Lemma \ref{lma:unitaryeq} for every $a$, the state $\psi_a=\psi\circ\beta_a$ is unitarily equivalent to $\psi$. Let us fix vectors $|\psi_a\ral\in\CH_\psi$  representing the states $\psi \circ \beta_a$. For each $a$, there is a unique family of unitary operators  $S_a:U_{1,0} \ra \BoundedOper(\CH_{\psi})$ representing $\beta_a$ such that $S_a |\psi_a\ral = |\psi\ral$. Moreover, since $\beta:U_{1,0}\ra \alLPA$ is strongly continuous, each $S_a$ is also strongly continuous.  The same arguments as in the proof of Theorem \ref{thm:gLSM} show that strongly continuous families of unitary operators $W^{(g_1)}_{ab} := S^{(g_1)}_a \left(S^{(g_1)}_b\right)^{-1}:U_{1,1} \to \BoundedOper(\CH_{\psi})$ and $V^{(g_1,g_2)}_{a} := S^{(g_1)}_a S^{(g_2)}_a \left(S^{(g_1 g_2)}_a\right)^{-1}:U_{2,0} \to \BoundedOper(\CH_{\psi})$ are $\pi_{\psi}$-images of smooth families of unitary observables $\CW^{(g_1)}_{ab}$, $\CV^{(g_1,g_2)}_{a}$ times not-necessarily-smooth functions $U_{1,1}\ra U(1)$ and $U_{2,0}\ra U(1)$. To complete the proof, it is sufficient show that these smooth families of observables can be chosen so that the functions are all $1$. 

Let $\CD_{\omega} \subset \CH_{\psi}$ be the subspace spanned by vectors of the form $\pi_\psi(\CA)|\omega\ral$, $\CA\in\SAal$ for $\omega = \psi$ or $\omega = \psi_a$. Note that since $W^{(g_1)}_{ab}$ at each point of $U_{1,1}$ is a conjugation with an almost local unitary, $\CD_{\psi_a}$ does not depend on $a$. Let $\derF_a = \beta^{-1}_a d \beta_a$. It does not excite $\psi_a$. By Remark \ref{rem:derGNSideal}, for every $u\in U^{(1)}_a$, we can define an unbounded operator $F_a(u):\CD_{\psi_a} \ra\CH_\psi\otimes_\RR T^*U^{(1)}_a\vert_u$ by
\beq\label{eq:Fadef}
F_a(u)\pi_{\psi}(\CA)|\psi_a\ral=\pi_{\psi}(\derF_a(\CA)\vert_u)|\psi_a\ral.
\eeq
These operators represent derivations $\derF_a$, in the sense that on $\CD_{\psi_a}$ we have an operator identity
\beq\label{eq:Farep}
[F_a(u),\pi_\psi(\CB)]=\pi_\psi(\derF_a(\CB)\vert_u),\quad\forall u\in U^{(1)}_a,\forall\CB\in\SAal.
\eeq
Also, $F_a(u)|\psi_a\ral=0$ for all $u\in U^{(1)}_a$ and all $a$. Further, by Lemma \ref{lma:nonuniquenesssmooth}, $\derF_a - \derF_{b}\in \Omega^1(U_{1,1},\mfkDal)$ is of the form $\ad_{\chh_{ab}}$, where $\chh_{ab}\in\Omega^1(U_{1,1},\mfkdal)$. Then  (\ref{eq:Fadef}) implies that $\forall u\in U^{(1)}_a\cap U^{(1)}_b$ the operator $F_a(u)-F_b(u)$ is bounded on $\CD_\psi$ and thus extends uniquely to $\CH_\psi$, while (\ref{eq:Farep}) implies that this extension is equal to $\pi_\psi(\chh_{ab}(u))$ plus a scalar function (i.e. plus an identity operator times an element of $T^*(U^{(1)}_a\cap U^{(1)}_b)\vert_u$ which {\it a priori} need not depend smoothly on $u$). 

To prove that $F_a(u)-F_b(u)$ is a $\pi_\psi$-image of an element of $\Omega^1(U^{(1)}_a\cap U^{(1)}_b,\mfkdal)$, it is sufficient to show that the 1-form $\lal\psi_b|F_a(u) - F_{b}(u) |\psi_b\ral$ smoothly depends on $u$. We have
\beq
\lal\psi_b|F_a(u) - F_{b}(u)|\psi_b \ral = \lal\psi_b|F_a(u) |\psi_b \ral = \lal\psi_a|S_a(u)^{-1} S_b(u) [F_a(u),S_b(u)^{-1} S_a(u)] |\psi_a\ral.
\eeq
The family of unitary operators $S_a(u)^{-1}S_b(u)$ is a (not-necessarily smooth) scalar function of $u$ times the $\pi_\psi$-image of a smooth family of unitary almost local observables. Since $F_a(u)$ represents a smooth family of derivations $\derF_a(u)$, the family of operators 
\beq\label{eq:operatorfamily}
S_a(u)^{-1} S_b(u) [F_a(u),S_b(u)^{-1} S_a(u)]
\eeq
is a $\pi_\psi$-image of an element of $\Omega^1(U^{(1)}_a\cap U^{(1)}_b,\mfkdal)$, which implies the claim.

Consider the strongly continuous family of unitary operators $W_{ab}(u)=S_a(u)S_b(u)^{-1}$. For any $\CA\in\SAal$, the function 
$u\mapsto S_a(u)\pi_\psi(\CA)|\psi_a\ral=\pi_\psi(\beta_a(u)(\CA))|\psi\ral$ is a smooth function $U^{(1)}_a\ra\CD_\psi$ and its differential is equal to 
\beq
d(S_a\pi_\psi(\CA)|\psi_a\ral)=S_a\pi_\psi(\derF_a(\CA))|\psi_a\ral=S_a F_a\pi_\psi(\CA)|\psi_a\ral
\eeq
Similarly, for any $\CA\in\SAal$, the function 
$u\mapsto S_b^{-1}(u)\pi_\psi(\CA)|\psi\ral=\pi_\psi(\beta_b(u)^{-1}(\CA))|\psi_b\ral$ is a smooth function $U^{(1)}_b\ra \CD_{\psi_b}$, and its differential is 
\beq
d(S_b^{-1}\pi_\psi(\CA)|\psi\ral)=-F_b S_b^{-1}\pi_\psi(\CA)|\psi\ral.
\eeq
Therefore, for any
$\CA\in\SAal$, the function $W_{ab}\pi_\psi(\CA)|\psi\ral:U^{(1)}_a\cap U^{(1)}_b\ra\CD_\psi$ is smooth and its differential is
\beq\label{eq:Weq}
d\left(W_{ab}\pi_\psi(\CA)|\psi\ral\right)=W_{ab}S_b(F_a-F_b)S_b^{-1}\pi_\psi(\CA)|\psi\ral.
\eeq
On the other hand, we showed above that $F_a-F_b$ and therefore $S_b(F_a-F_b)S_b^{-1}$ is a $\pi_\psi$-image of an element $\chf_{ab}\in\Omega^1(U^{(1)}_a\cap U^{(1)}_b,\mfkdal)$. By Lemma \ref{lma:inner}, the equation
\beq
d\CW_{ab}=\CW_{ab}\chf_{ab}
\eeq
has a smooth solution $\CW_{ab}:U^{(1)}_a\cap U^{(1)}_b\ra\SUal$ for any choice of $\CW_{ab}(u_0)$, where $u_0\in U^{(1)}_a\cap U^{(1)}_b$ is an arbitrary basepoint. If we pick $\CW_{ab}(u_0)$ so that $\pi_\psi(\CW_{ab}(u_0))=W_{ab}(u_0)$, then $\pi_\psi(\CW_{ab})$ solves (\ref{eq:Weq}) and agrees with $W_{ab}$ at $u=u_0$. Therefore, $W_{ab}=\pi_\psi(\CW_{ab})$ everywhere on $U^{(1)}_a\cap U^{(1)}_b.$

Similarly, we can consider the family of derivation $\beta^{(g_1 g_2)}(\beta^{(g_2)-1} (\derF^{(g_1)}_a)) + \beta^{(g_1 g_2)}(\derF^{(g_2)}-\derF^{(g_1 g_2)})$ to argue that the family of operators
\beq
S^{(g_1 g_2)}_a S^{(g_2)-1}_a F^{(g_1)}_a S^{(g_2)}_a S^{(g_1 g_2)-1}_a + S^{(g_1 g_2)}_a (F^{(g_2)}_a-F^{(g_1 g_2)}_a) S^{(g_1 g_2)-1}_a
\eeq
represents a smooth family of almost local observables. It implies that $V^{(g_1,g_2)}_{a}$ is the $\pi_{\psi}$-image of a smooth family $\CV:U_{2,0} \to \SAal$.

The case when the image of $\alpha$ contains automorphisms with a non-trivial GNVW index can be treated in the same way as in the proof of Theorem \ref{thm:gLSM}.

\end{proof}

\section{Some computations} \label{app:cocycle}

\subsection{Abstract group} \label{app:cocycleD}

We denote the components of the function $f$ on $G^p$ by $f^{(g_1,...,g_p)}$. With this notation, we have automorphisms $\beta^{(g_1)}$, unitary observables $\CV^{(g_1,g_2)}$ satisfying
\beq
\beta^{(g_1)} \circ \beta^{(g_2)} = \Ad_{\CV^{(g_1,g_2)}} \circ \beta^{(g_1 g_2)} \\
\eeq
and $U(1)$-valued functions $\omega^{(g_1,g_2,g_3)}$ defined by
\bqa
\omega^{(g_1,g_2,g_3)} = \CV^{(g_1,g_2)} \CV^{(g_1 g_2, g_3)} \CV^{(g_1,g_2 g_3)-1} \beta^{(g_1)} (\CV^{(g_2,g_3)})^{-1}.
\eqa

\subsubsection{Cocycle condition} \label{app:CocycleConditionD}

The cocycle condition for $\omega$ follows from
\begin{multline}
\omega^{(g_1g_2,g_3,g_4)} \omega^{(g_1,g_2,g_3g_4)}  (\Ad_{\CV^{(g_1,g_2)}} \circ \beta^{(g_1 g_2)})(\CV^{(g_3,g_4)}) \beta^{(g_1)}(\CV^{(g_2,g_3 g_4)}) \CV^{(g_1,g_2 g_3 g_4)} = \\
\CV^{(g_1,g_2)} \omega^{(g_1g_2,g_3,g_4)} \beta^{(g_1 g_2)}(\CV^{(g_3,g_4)}) \CV^{(g_1,g_2)-1} \omega^{(g_1,g_2,g_3g_4)} \beta^{(g_1)}(\CV^{(g_2,g_3 g_4)}) \CV^{(g_1,g_2 g_3 g_4)} = \\
=  \CV^{(g_1,g_2)} \omega^{(g_1g_2,g_3,g_4)} \beta^{(g_1 g_2)}(\CV^{(g_3, g_4)}) \CV^{(g_1 g_2, g_3 g_4)} = \\
\CV^{(g_1,g_2)} \CV^{(g_1 g_2, g_3)} \CV^{(g_1 g_2 g_3, g_4)} = \\
= \omega^{(g_1,g_2,g_3)} \beta^{(g_1)}(\CV^{(g_2,g_3)}) \CV^{(g_1, g_2 g_3)} \CV^{(g_1 g_2 g_3,g_4)} = \\
\omega^{(g_1,g_2,g_3)} \omega^{(g_1,g_2 g_3,g_4)} \beta^{(g_1)}(\CV^{(g_2,g_3)}) \beta^{(g_1)}(\CV^{(g_2 g_3,g_4)}) \CV^{(g_1, g_2 g_3 g_4)} = \\
= \omega^{(g_1,g_2,g_3)} \omega^{(g_1,g_2 g_3,g_4)} \omega^{(g_2, g_3, g_4)} (\beta^{(g_1)} \circ \beta^{(g_2)})(\CV^{(g_3,g_4)}) \beta^{(g_1)}(\CV^{(g_2, g_3 g_4)}) \CV^{(g_1, g_2 g_3 g_4)} = \\
=\omega^{(g_1,g_2,g_3)} \omega^{(g_1,g_2 g_3,g_4)} \omega^{(g_2, g_3, g_4)} (\Ad_{\CV^{(g_1,g_2)}} \circ \beta^{(g_1 g_2)})(\CV^{(g_3,g_4)}) \beta^{(g_1)}(\CV^{(g_2, g_3 g_4)}) \CV^{(g_1, g_2 g_3 g_4)}.
\end{multline}

\subsubsection{Independence of the choice of $\beta$} \label{app:IndependenceDiscrete}

Two different choices of $\beta$ are related by $\tilde{\beta}^{(g)} = \beta^{(g)} \circ \Ad_{\CU^{(g)}}$ for some unitary observables $\CU^{(g)}$. It implies
$$
\tilde{\CV}^{(g_1,g_2)} = \beta^{(g_1)} (\CU^{(g_1)}) \CV^{(g_1,g_2)} \beta^{(g_1 g_2)}(\CU^{(g_2)} \CU^{(g_1 g_2)-1}),
$$
We have
\begin{multline}
\tilde{\omega}^{(g_1,g_2,g_3)} =  \tilde{\CV}^{g_1,g_2} \tilde{\CV}^{g_1g_2,g_3}   \tilde{\CV}^{(g_1,g_2 g_3)-1} \tilde{\beta}^{(g_1)} (\tilde{\CV}^{(g_2,g_3)})^{-1} = \\
\beta^{(g_1)} (\CU^{(g_1)}) \CV^{(g_1,g_2)} \beta^{(g_1 g_2)}(\CU^{(g_2)}) \CV^{(g_1 g_2,g_3)} \beta^{(g_1 g_2 g_3)}(\CU^{(g_3)} \CU^{(g_2 g_3)-1}) \times \\ \times \CV^{(g_1,g_2 g_3)-1} \beta^{(g_1)}(\beta^{(g_2 g_3)}(\CU^{(g_2 g_3)} \CU^{(g_3)-1})) \beta^{(g_1)}(\CV^{(g_2,g_3)-1}) \beta^{(g_1)}(\CU^{(g_1)} \beta^{(g_2)}(\CU^{(g_2)}))^{-1} = \\ 
\beta^{(g_1)}(\CU^{(g_1)} \beta^{(g_2)}(\CU^{(g_2)})) \CV^{(g_1,g_2)} \CV^{(g_1 g_2,g_3)} \beta^{(g_1 g_2 g_3)}(\CU^{(g_3)} \CU^{(g_2 g_3)-1}) \times \\ \times \CV^{(g_1,g_2 g_3)-1} \beta^{(g_1)}(\beta^{(g_2 g_3)}(\CU^{(g_2 g_3)} \CU^{(g_3)-1})) \beta^{(g_1)}(\CV^{(g_2,g_3)-1}) \beta^{(g_1)}(\CU^{(g_1)} \beta^{(g_2)}(\CU^{(g_2)}))^{-1} = \\
\beta^{(g_1)}(\CU^{(g_1)} \beta^{(g_2)}(\CU^{(g_2)})) \CV^{(g_1,g_2)} \CV^{(g_1 g_2,g_3)} \CV^{(g_1,g_2 g_3)-1} \beta^{(g_1)}(\CV^{(g_2,g_3)-1}) \beta^{(g_1)}(\CU^{(g_1)} \beta^{(g_2)}(\CU^{(g_2)}))^{-1} = \\
= \omega^{(g_1,g_2,g_3)}
\end{multline}

\subsection{Lie group}  \label{app:cocycleL}

In the case of a Lie group $G$ and some good simplicial cover $U_{\bullet,\bullet}$ of $B_{\bullet} G$, it is convenient to use the coordinates and the labels for each chart. We denote the components of the function $f$ on $U^{(p)}_{a_1} \cap ... \cap U^{(p)}_{a_q} \subset G^p$ by $f^{(g_1,...,g_p)}_{a_1...a_q}$. With this notation, we have automorphisms $\beta^{(g_1)}_{a}$, unitary observables $\CV^{(g_1,g_2)}_{a}$ and $\CW^{(g_1)}_{ab}$ satisfying
\bqa
\beta^{(g_1)}_{a} = \Ad_{\CW^{(g_1)}_{ab}} \circ \beta^{(g_1)}_{b} \\
\beta^{(g_1)}_{a} \circ \beta^{(g_2)}_{a} = \Ad_{\CV^{(g_1,g_2)}_{a}} \circ \beta^{(g_1 g_2)}_{a} \\
\eqa
and $U(1)$-valued functions $\omega^{(g_1,g_2,g_3)}_a$, $\ups^{(g_1,g_2)}_{ab}$, $\eta^{(g_1)}_{abc}$ defined by
\bqa
\omega^{(g_1,g_2,g_3)}_a = \CV^{(g_1,g_2)}_{a} \CV^{(g_1 g_2, g_3)}_{a} \CV^{(g_1,g_2 g_3)-1}_{a} \beta^{(g_1)}_a (\CV^{(g_2,g_3)}_{a})^{-1},\\
\ups^{(g_1,g_2)}_{ab} = \CV^{(g_1,g_2)}_{a} \CW^{(g_1 g_2)}_{a b} \CV^{(g_1,g_2)-1}_{b} \beta^{(g_1)}_{b}(\CW^{(g_2)}_{ab})^{-1} \CW^{(g_1)-1}_{ab},\\
\eta^{(g_1)}_{abc} = \CW^{(g_1)}_{ab} \CW^{(g_1)}_{bc} \CW^{(g_1)-1}_{ac}.
\eqa

\subsubsection{Cocycle condition} \label{app:CocycleConditionL}

To show the cocycle condition for $(\omega,\ups,\eta)$ we need to show $\delta \eta = (\delta \ups)(d \eta) = (\delta \omega)(d \ups)^{-1} = d \omega = 1$. We have 
\begin{multline}
\eta^{(g_1)}_{abc} \eta^{(g_1)}_{acd} = (\CW^{(g_1)}_{ab} \CW^{(g_1)}_{bc} \CW^{(g_1)-1}_{ac}) (\CW^{(g_1)}_{ac} \CW^{(g_1)}_{cd} \CW^{(g_1)-1}_{ad}) = \\ = \CW^{(g_1)}_{ab} (\CW^{(g_1)}_{bc} \CW^{(g_1)}_{cd} \CW^{(g_1)-1}_{bd}) \CW^{(g_1)-1}_{ab} (\CW^{(g_1)}_{ab} \CW^{(g_1)}_{bd} \CW^{(g_1)-1}_{ad}) = \\ =  \CW^{(g_1)}_{ab} \eta^{(g_1)}_{bcd} \CW^{(g_1)-1}_{ab} \eta_{abd} = \eta^{(g_1)}_{bcd} \eta^{(g_1)}_{abd}
\end{multline}
that implies $\delta \eta = 1$. The identity
\begin{multline}
(\ups^{(g_1,g_2)}_{ab}) (\eta^{(g_1)}_{abc} ) (\ups^{(g_1,g_2)-1}_{ac}) \CV^{(g_1,g_2)}_{a}(\eta^{(g_1 g_2)-1}_{abc}) \CW^{(g_1 g_2)}_{ab} \CV^{(g_1,g_2)-1}_{b} ( \ups^{(g_1,g_2)-1}_{bc} )  = \\ 
= \CV^{(g_1,g_2)}_{a} \CW^{(g_1 g_2)}_{ab} \CV^{(g_1,g_2)-1}_{b} \l \beta^{(g_1)}_{b} (\CW^{(g_2)}_{ab})^{-1} \CW^{(g_1)}_{bc} \beta^{(g_1)}_{c} (\CW^{(g_2)}_{ac} \CW^{(g_2)-1}_{bc})  \CW^{(g_1)-1}_{bc} \r = \\
= \CV^{(g_1,g_2)}_{a} \CW^{(g_1 g_2)}_{ab} \CV^{(g_1,g_2)-1}_{b} \l\beta^{(g_1)}_{b} (\CW^{(g_2)-1}_{ab}\eta^{(g_2)-1}_{abc} \CW^{(g_2)}_{ab}) \r   
\end{multline}
implies $(\delta \ups) (d \eta) = 1$. To show $d \ups = \delta \omega$, we use
\begin{multline}
(\ups^{(g_1,g_2)-1}_{ab}\CV^{(g_1,g_2)}_{a})(\ups^{(g_1 g_2,g_3)-1}_{ab} \CV^{(g_1 g_2, g_3)}_{a}) (\CV^{(g_1,g_2 g_3)-1}_{a} \ups^{(g_1 g_2,g_3)}_{ab}) = \\
= \CW^{(g_1)}_{ab} \beta^{(g_1)}_{b}(\CW^{(g_2)}_{ab}) \CV^{(g_1,g_2)}_b \beta^{(g_1 g_2)}_b(\CW^{(g_3)}_{ab}) \CV^{(g_1 g_2,g_3)}_{b} \CV^{(g_1,g_2 g_3)-1}_{b} \beta^{(g_1)}_b(\CW^{(g_2 g_3)-1}_{ab}) \CW^{(g_1)-1}_{ab} = \\
= \CW^{(g_1)}_{ab} \beta^{(g_1)}_{b}(\CW^{(g_2)}_{ab}) \beta^{(g_1)}_{b} (\beta^{( g_2)}_b(\CW^{(g_3)}_{ab})) \CV^{(g_1,g_2)}_b \CV^{(g_1 g_2,g_3)}_{b} \CV^{(g_1,g_2 g_3)-1}_{b} \beta^{(g_1)}_b(\CW^{(g_2 g_3)-1}_{ab}) \CW^{(g_1)-1}_{ab} = \\
= \omega^{(g_1,g_2,g_3)}_{b} \CW^{(g_1)}_{ab} \beta^{(g_1)}_{b}(\CW^{(g_2)}_{ab} \beta^{( g_2)}_b(\CW^{(g_3)}_{ab}) \CV^{(g_2,g_3)}_b \CW^{(g_2 g_3)-1}_{ab} ) \CW^{(g_1)-1}_{ab} = \\ = \omega^{(g_1,g_2,g_3)}_{b}
 \beta^{(g_1)}_{a} (\ups^{(g_2, g_3)-1}_{ab} \CV^{(g_2,g_3)}_a) 
\end{multline}
Finally, to show $d \omega = 1$, we use
\begin{multline}
\omega_a^{(g_1g_2,g_3,g_4)} \omega_a^{(g_1,g_2,g_3g_4)}  (\Ad_{\CV_a^{(g_1,g_2)}} \circ \beta_a^{(g_1 g_2)})(\CV_a^{(g_3,g_4)}) \beta_a^{(g_1)}(\CV_a^{(g_2,g_3 g_4)}) \CV_a^{(g_1,g_2 g_3 g_4)} = \\
\CV_a^{(g_1,g_2)} \omega_a^{(g_1g_2,g_3,g_4)} \beta_a^{(g_1 g_2)}(\CV_a^{(g_3,g_4)}) \CV_a^{(g_1,g_2)-1} \omega_a^{(g_1,g_2,g_3g_4)} \beta_a^{(g_1)}(\CV_a^{(g_2,g_3 g_4)}) \CV_a^{(g_1,g_2 g_3 g_4)} = \\
=  \CV_a^{(g_1,g_2)} \omega_a^{(g_1g_2,g_3,g_4)} \beta_a^{(g_1 g_2)}(\CV_a^{(g_3, g_4)}) \CV_a^{(g_1 g_2, g_3 g_4)} = \\
\CV_a^{(g_1,g_2)} \CV_a^{(g_1 g_2, g_3)} \CV_a^{(g_1 g_2 g_3, g_4)} = \\
= \omega_a^{(g_1,g_2,g_3)} \beta_a^{(g_1)}(\CV_a^{(g_2,g_3)}) \CV_a^{(g_1, g_2 g_3)} \CV_a^{(g_1 g_2 g_3,g_4)} = \\
\omega_a^{(g_1,g_2,g_3)} \omega_a^{(g_1,g_2 g_3,g_4)} \beta_a^{(g_1)}(\CV_a^{(g_2,g_3)}) \beta_a^{(g_1)}(\CV_a^{(g_2 g_3,g_4)}) \CV_a^{(g_1, g_2 g_3 g_4)} = \\
= \omega_a^{(g_1,g_2,g_3)} \omega_a^{(g_1,g_2 g_3,g_4)} \omega_a^{(g_2, g_3, g_4)} (\beta_a^{(g_1)} \circ \beta_a^{(g_2)})(\CV_a^{(g_3,g_4)}) \beta_a^{(g_1)}(\CV_a^{(g_2, g_3 g_4)}) \CV_a^{(g_1, g_2 g_3 g_4)} = \\
=\omega_a^{(g_1,g_2,g_3)} \omega_a^{(g_1,g_2 g_3,g_4)} \omega_a^{(g_2, g_3, g_4)} (\Ad_{\CV_a^{(g_1,g_2)}} \circ \beta_a^{(g_1 g_2)})(\CV_a^{(g_3,g_4)}) \beta_a^{(g_1)}(\CV_a^{(g_2, g_3 g_4)}) \CV_a^{(g_1, g_2 g_3 g_4)}.
\end{multline}

\subsubsection{Independence of the choice of $\beta$} \label{app:IndependenceOfBetaL}

Two different choices of $\beta$ are related by $\tilde{\beta}^{(g)}_a = \beta^{(g)}_a \circ \Ad_{\CU^{(g)}_a}$ for some smooth family of unitary almost local observables $\CU^{(g)}_a$. It implies
$$
\tilde{\CV}^{(g_1,g_2)}_a = \beta_a^{(g_1)} (\CU^{(g_1)}) \CV^{(g_1,g_2)}_a \beta^{(g_1 g_2)}_a(\CU^{(g_2)} \CU^{(g_1 g_2)-1}),
$$
$$
\tilde{\CW}^{(g_1)}_{ab} = \CW^{(g_1)}_{ab} \beta^{(g_1 )}_b(\CU^{(g_1)}\CU^{(g_1)-1}_{b}).
$$
We have
\begin{multline}
\tilde{\eta}^{(g_1)}_{abc} =  \tilde{\CW}^{(g_1)}_{ab} \tilde{\CW}^{(g_1)}_{bc} \tilde{\CW}^{(g_1)-1}_{ac} = \CW^{(g_1)}_{ab} \beta^{(g_1 )}_b(\CU^{(g_1)}\CU^{(g_1)-1}_{b}) \CW^{(g_1)}_{bc} \beta^{(g_1 )}_c(\CU^{(g_1)}_{b}\CU^{(g_1)-1}) \CW^{(g_1)-1}_{ac} = \\
= \CW^{(g_1)}_{ab} \CW^{(g_1)}_{bc} \beta^{(g_1 )}_c(\CU^{(g_1)}\CU^{(g_1)-1}_{b})  \beta^{(g_1 )}_c(\CU^{(g_1)}_{b}\CU^{(g_1)-1}) \CW^{(g_1)-1}_{ac} = \eta^{(g_1)}_{abc},\\
\end{multline}
\begin{multline}
\tilde{\ups}^{(g_1,g_2)}_{ab} = \beta^{(g_1)}_{a} (\CU^{(g_1)}_{a}) \CV^{(g_1,g_2)}_a \beta^{(g_1 g_2)}_a(\CU^{(g_2)}_{a} \CU^{(g_1 g_2)-1}_{a})) \CW^{(g_1 g_2)}_{ab} \beta^{(g_1 g_2)}_b(\CU^{(g_1g_2)}_{a} \CU^{(g_2)-1}_{b}) \CV^{(g_1,g_2)-1}_b \times \\ \times \beta^{(g_1)}_{b} (\beta^{(g_2)}_{b}(\CU^{(g_2)}_{b} \CU^{(g_2)-1}_{a})\CW^{(g_2)-1}_{ab}) \beta^{(g_1)}_b(\CU^{(g_1)-1}_{a}) \CW^{(g_1)-1}_{ab} =  \\ =
\beta^{(g_1)}_{a} (\CU^{(g_1)}_{a}) \CV^{(g_1,g_2)}_a \CW^{(g_1 g_2)}_{ab} \beta^{(g_1 g_2)}_b (\CU^{(g_2)}_{a} \CU^{(g_2)-1}_{b}) \CV^{(g_1,g_2)-1}_b \times \\ \times \beta^{(g_1)}_{b} (\beta^{(g_2)}_{b}(\CU^{(g_2)}_{b} \CU^{(g_2)-1}_{a})\CW^{(g_2)-1}_{ab}) \beta^{(g_1)}_b(\CU^{(g_1)-1}_{a}) \CW^{(g_1)-1}_{ab} = \\ 
= \beta^{(g_1)}_{a} (\CU^{(g_1)}_{a}) \CV^{(g_1,g_2)}_a \CW^{(g_1 g_2)}_{ab} \CV^{(g_1,g_2)-1}_b \beta^{(g_1)}_{b} (\CW^{(g_2)-1}_{ab}) \beta^{(g_1)}_b(\CU^{(g_1)-1}_{a}) \CW^{(g_1)-1}_{ab} = \\
= \beta^{(g_1)}_{a} (\CU^{(g_1)}_{a}) \CV^{(g_1,g_2)}_a \CW^{(g_1 g_2)}_{ab} \CV^{(g_1,g_2)-1}_b \beta^{(g_1)}_{b} (\CW^{(g_2)-1}_{ab}) \CW^{(g_1)-1}_{ab} \beta^{(g_1)}_a(\CU^{(g_1)-1}_{a}) = \ups^{(g_1,g_2)}_{ab},
\end{multline}

\begin{multline}
\tilde{\omega}^{(g_1,g_2,g_3)}_{a} =  \tilde{\CV}^{g_1,g_2}_a \tilde{\CV}^{g_1g_2,g_3}_a   \tilde{\CV}^{(g_1,g_2 g_3)-1}_a \tilde{\beta}^{(g_1)}_{a} (\tilde{\CV}^{(g_2,g_3)}_a)^{-1} = \\
\beta^{(g_1)}_{a} (\CU^{(g_1)}_{a}) \CV^{(g_1,g_2)}_a \beta^{(g_1 g_2)}_a(\CU^{(g_2)}_{a}) \CV^{(g_1 g_2,g_3)}_a \beta^{(g_1 g_2 g_3)}_a(\CU^{(g_3)}_{a} \CU^{(g_2 g_3)-1}_{a}) \times \\ \times \CV^{(g_1,g_2 g_3)-1}_a \beta^{(g_1)}_a(\beta^{(g_2 g_3)}_a(\CU^{(g_2 g_3)}_{a} \CU^{(g_3)-1}_{a})) \beta^{(g_1)}_{a}(\CV^{(g_2,g_3)-1}_a) \beta^{(g_1)}_a(\CU^{(g_1)}_a \beta^{(g_2)}_a(\CU^{(g_2)}_{a}))^{-1} = \\ 
\beta^{(g_1)}_a(\CU^{(g_1)}_a \beta^{(g_2)}_a(\CU^{(g_2)}_{a})) \CV^{(g_1,g_2)}_a \CV^{(g_1 g_2,g_3)}_a \beta^{(g_1 g_2 g_3)}_a(\CU^{(g_3)}_{a} \CU^{(g_2 g_3)-1}_{a}) \times \\ \times \CV^{(g_1,g_2 g_3)-1}_a \beta^{(g_1)}_a(\beta^{(g_2 g_3)}_a(\CU^{(g_2 g_3)}_{a} \CU^{(g_3)-1}_{a})) \beta^{(g_1)}_{a}(\CV^{(g_2,g_3)-1}_a) \beta^{(g_1)}_a(\CU^{(g_1)}_a \beta^{(g_2)}_a(\CU^{(g_2)}_{a}))^{-1} = \\
\beta^{(g_1)}_a(\CU^{(g_1)}_a \beta^{(g_2)}_a(\CU^{(g_2)}_{a})) \CV^{(g_1,g_2)}_a \CV^{(g_1 g_2,g_3)}_a \CV^{(g_1,g_2 g_3)-1}_a \beta^{(g_1)}_{a}(\CV^{(g_2,g_3)-1}_a) \beta^{(g_1)}_a(\CU^{(g_1)}_a \beta^{(g_2)}_a(\CU^{(g_2)}_{a}))^{-1} = \\
= \omega^{(g_1,g_2,g_3)}_{a}
\end{multline}


\printbibliography

\end{document}